\documentclass[aps,prl,reprint,superscriptaddress,nofootinbib,longbibliography]{revtex4-2}

\usepackage[T1]{fontenc}
\usepackage[utf8]{inputenc}
\usepackage{amsmath,amssymb,amsthm,mathtools,bm}
\usepackage{booktabs}
\usepackage{graphicx}
\usepackage{float}
\usepackage{xcolor}
\usepackage{hyperref}

\graphicspath{{figures/}}

\hypersetup{
    colorlinks=true,
    linkcolor=blue!50!black,
    citecolor=blue!50!black,
    urlcolor=blue!50!black
}

\newcommand{\ii}{\mathrm{i}}
\newcommand{\Tr}{\operatorname{Tr}}
\newcommand{\id}{\mathbb I}
\newcommand{\ket}[1]{|#1\rangle}
\newcommand{\bra}[1]{\langle #1|}
\newcommand{\proj}[1]{\ket{#1}\!\bra{#1}}

\newcommand{\cJ}{\mathcal J}
\newcommand{\cB}{\mathcal B}
\newcommand{\cO}{\mathcal O}
\newcommand{\GF}{\mathbb F_2}
\newcommand{\symp}[2]{\langle #1,#2\rangle_{\rm sp}}

\newtheorem{lemma}{Lemma}
\newtheorem{proposition}{Proposition}
\newtheorem{theorem}{Theorem}
\newtheorem{corollary}{Corollary}

\begin{document}

\title{Fisher-Orthogonal Memory in Quantum Reservoir Computing}

\author{Ce Wang}
\email{phywangce@gmail.com}
\affiliation{School of Physics Science and Engineering, Tongji University, Shanghai 200092, China}
\author{Xingze Qiu}
\affiliation{School of Physics Science and Engineering, Tongji University, Shanghai 200092, China}

\date{\today}

\begin{abstract}
Quantum reservoir computing processes temporal information through driven
many-body dynamics, but its performance is ultimately limited by how accurately
past inputs can be extracted from finite measurements. Here we formulate this
limitation as a local multiparameter estimation problem and introduce a
delay-space quantum Fisher information matrix to quantify the distinguishability
of information stored at different delays.
This perspective identifies Fisher-orthogonal memory as a
measurement-efficient design principle: different delays should perturb the
reservoir state along mutually
Fisher-orthogonal directions. We first analyze the
single-qubit limit using the Gill--Massar bound, revealing an optimal
write-store-routing trade-off. Guided by this structure, we construct solvable
multi-qubit reservoirs based on Clifford routing orbits and Singer-cycle Pauli
algebra. The resulting dynamics yield diagonal delay-space QFIMs with
analytically programmable fading profiles. Under finite-shot
local Pauli readout, these reservoirs retain sharp memory windows
that are absent in a
validation-selected Ising baseline. Their product-task behavior is governed by
second-order responses
inherited from the same Pauli-routing algebra.
Our results provide an analytically controlled route toward
measurement-efficient quantum reservoir computing.
\end{abstract}

\maketitle

\textit{Introduction.---}
Quantum reservoir computing (QRC) processes sequential data through repeated
input injection, fixed reservoir evolution, and classical readout
\cite{Fujii2017Disordered,Ghosh2019QuantumReservoirProcessing,Mujal2021Opportunities}. Across this broad landscape, diverse reservoir architectures
\cite{Nakajima2019SpatialMultiplexing,ChenNurdin2019DissipativeMaps,Nokkala2021Gaussian,MartinezPena2023FiniteDimensions,Sannia2024Dissipation,Kubota2023Noise,MartinezPena2025InputDependence,Schutte2025Expressivity,Cindrak2026MemoryNonlinearity}
have been actively explored, ranging from continuous-time 
spin ensembles
\cite{Fujii2017Disordered,Negoro2018NMR,Suzuki2022Natural,MartinezPena2023SpinIPC,Hou2026CorrelatedSpins}
and weak-measurement or feedback networks
\cite{Mujal2023WeakProjectiveMeasurements,Yasuda2023RepeatedMeasurements,Kobayashi2024FeedbackDriven,Monomi2025WeakMeasurements,Zhu2025FeedbackQRC,Franceschetto2026Backaction}
to discrete quantum circuits
\cite{Chen2020NoisyQuantumComputers,Xia2023ConfiguredQRC,Connerty2026NISQEcho}, 
with successful experimental demonstrations realized across nuclear magnetic resonance
\cite{Negoro2018NMR,Hou2026CorrelatedSpins}, 
superconducting qubits
\cite{Chen2020NoisyQuantumComputers,Suzuki2022Natural,Yasuda2023RepeatedMeasurements}, 
and photonic platforms
\cite{Paparelle2026OpticalQRC}. Concurrently, intense theoretical interest 
has focused on linking QRC performance to fundamental many-body physics, highlighting 
the roles of quantum entanglement
\cite{Kora2024EntanglementQRC,Karimi2025EntanglementKerrQRC,Askari2025DistributedEntanglementQRC}, 
operator spreading and scrambling
\cite{Nahum2018OperatorSpreading,Xiong2025ScramblingNoiseQRC}, quantum chaos and 
symmetry-induced concentration
\cite{Kobayashi2026ChaosQRC,Sannia2025ConcentrationQRC}, and 
non-stabilizer resources such as magic
\cite{Ivaki2026Universality,Xia2026MagicQRC}. Despite these advances, standard evaluations usually assume exact expectation
values and therefore bypass a basic experimental constraint because information
must be inferred from finitely many shots. A useful finite-shot reservoir needs not
only memory capacity but \emph{identifiable} memory. If different historical
inputs perturb the final state along nearly parallel directions, projection
noise obscures their temporal origin despite large individual sensitivities.

Here, we address this challenge by formulating QRC memory as a local
multiparameter-estimation problem and introducing a delay-space quantum Fisher
information matrix (QFIM) to quantify how well different past inputs can be
resolved under finite measurement resources. Guided by the single-qubit optimum,
we construct a multi-qubit reservoir based on a solvable Clifford routing orbit
(CRO). In this design, historical inputs at different delays are routed to
orthogonal Pauli pathways, yielding a transparent mechanism for suppressing
finite-shot mixing between delay coordinates. The construction exposes an
explicit algebraic link between routing, dissipation, and memory
identifiability. Under local finite-shot Pauli readout, the resulting reservoirs
retain sharply resolved memory windows relative to a
validation-selected Ising baseline, demonstrating measurement-efficient
temporal processing.

\textit{QRC as a multiparameter estimation problem.---}
Following the standard discrete-time framework of QRC, each incoming data 
point is encoded into a fresh input qubit that interacts with the reservoir 
through a fixed recurrence unitary
\cite{Fujii2017Disordered}. This input qubit is then traced out 
before the subsequent signal arrives.  The classical signal $s$ is encoded via the state preparation:
\begin{equation}
    \ket{\psi_s}
    =
    \cos\!\left(\frac{\pi s}{2}\right)\ket{0}
    +
    \sin\!\left(\frac{\pi s}{2}\right)\ket{1},
    \qquad
    \tau_s=\proj{\psi_s}.
    \label{eq:input_encoding}
\end{equation}
The input qubit then interacts with an \(N\)-qubit reservoir through a fixed
unitary \(U\).  After tracing out the input, one driven reservoir step is the
channel
\begin{equation}
    \Phi_s(\rho)
    =
    \Tr_{\rm in}
    \left[
        U(\tau_s\otimes\rho)U^\dagger
    \right].
    \label{eq:driven_channel}
\end{equation}
Repeated use of this input-driven channel maps a recent input history into the
reservoir state measured by the readout.  To quantify
local memory distinguishability under
finite-measurement constraints, we 
examine how cleanly infinitesimal variations in the input history can be 
resolved within the driven steady state. We consider a reference background of 
a constant drive $s_0=1/2$, letting $\rho_\ast$ be the corresponding fixed 
point of the channel $\Phi_{s_0}$. Ordering the input stream backward from the 
moment of measurement, $\bm{s}=(s_1,\ldots,s_L)$ where $s_1$ represents the 
most recent signal, the final reservoir state evolves as
\begin{equation}
    \rho_f(\bm{s})
    =
    \Phi_{s_1}\circ\Phi_{s_2}\circ\cdots\circ\Phi_{s_L}(\rho_\ast).
    \label{eq:history_state}
\end{equation}
The final reservoir state is therefore parameterized by the
recent input history. In this local regime, reconstructing the past becomes a
multiparameter estimation problem for the perturbations
$\delta s_k=s_k-s_0$, with $\rho_f(\bm{s})$ as the readout state.

To first order in these perturbations, the local memory content is governed 
by the linear response of the final density matrix,
\begin{equation}
    \delta\rho_f
    =
    \sum_{k=1}^{L}\Delta_k\,\delta s_k,
    \qquad
    \Delta_k
    =
    \left.
    \frac{\partial\rho_f}{\partial s_k}
    \right|_{\bm s=s_0\bm{1}}
    =
    \Phi_{s_0}^{k-1}
    \left[
        \dot\Phi_{s_0}(\rho_\ast)
    \right],
    \label{eq:Delta_def}
\end{equation}
where $\dot\Phi_{s_0}=\partial_s\Phi_s|_{s_0}$. Each $\Delta_k$ is a traceless 
Hermitian tangent operator on the
reservoir. Together, the set $\{\Delta_k\}$ spans the local response space
associated with the input history. To quantify how effectively these
response modes can be distinguished under finite-measurement 
noise, we construct the QFIM in
delay space as
\begin{equation}
   (\mathsf H)_{kl}
= g_{\rho_\ast}(\Delta_k,\Delta_l).
    \label{eq:H_def}
\end{equation}
When $\rho$ is full rank, the SLD metric is
\(g_\rho(A,B)=\operatorname{Re}\Tr[A\,\mathcal J_\rho^{-1}(B)]\), with
\(\mathcal J_\rho(X)=(\rho X+X\rho)/2\)
\cite{Helstrom1976QuantumDetection,Braunstein1994StatisticalDistance,Paris2009QuantumEstimation,Szczykulska2016Multiparameter,Ragy2016Compatibility,SupplementalMaterial}. Each $\Delta_k$ belongs to the reservoir operator space, whereas
$\mathsf H$ acts on the $L$ delay coordinates. Its diagonal entries are the
single-delay QFIs, and its off-diagonal entries quantify Fisher overlap. Two
delays are Fisher-orthogonal when $(\mathsf H)_{kl}=0$ for $k\ne l$. If
$\mathsf H$ is rank deficient, some linear combination of historical
perturbations produces no first-order change in the reservoir state and is
therefore locally invisible under any measurement.

\textit{Single-qubit Fisher cycle.---}
For a single-qubit reservoir, the three-dimensional operator space restricts
the number of mutually Fisher-orthogonal directions to at most three.
Accordingly, we consider reconstructing the three most recent perturbations
($L=3$). A positive matrix $\Omega$ assigns their relative loss weights, and
we begin with the equal-delay choice
$\Omega_{\rm eq}=\operatorname{diag}(1/3,1/3,1/3)$. To quantify
multiparameter estimation precision, we use the Gill--Massar (GM) bound
\cite{Gill2000StateEstimation,SupplementalMaterial}
\begin{equation}
    C_{\rm GM}
    =
    \left\{
    \Tr
    \left[
        \left(
        \mathsf{H}^{-1/2}\Omega \mathsf{H}^{-1/2}
        \right)^{1/2}
    \right]
    \right\}^2 .
    \label{eq:GM_single}
\end{equation}
When measurements are
performed independently on each copy, $C_{\rm GM}$ bounds the weighted
estimation error, with lower values corresponding to higher precision. Guided
by this cost, we construct a tunable joint input--reservoir unitary of the
form
\begin{equation}
    \begin{aligned}
    U_r
    &=
    (\mathbb I_{\rm in}\otimes V_{\rm cyc})\,
    W_r\,
    (L_{\rm in}\otimes\mathbb I_{\rm R}),\\
    L_{\rm in}&=e^{+\ii\pi Y/4},
    \qquad
    V_{\rm cyc}=H_{\rm Had}S^\dagger,\\
    W_r
    &=
    \begin{pmatrix}
    a_r\mathbb I_{\rm R} & b_r Z\\
    b_r X & \ii a_r Y
    \end{pmatrix}_{\rm in}.
    \end{aligned}
    \label{eq:single_Ur}
\end{equation}
Here \(a_r=\sqrt{(1+r)/2}\), \(b_r=\sqrt{(1-r)/2}\), and
\(0<r<1\). The block form refers to the input basis, while its Pauli entries
act on the reservoir. With \(H_{\rm Had}\) and
\(S=\operatorname{diag}(1,\ii)\) denoting the Hadamard and phase gates,
respectively, \(V_{\rm cyc}\) generates
\(Z\mapsto X\mapsto Y\mapsto Z\).

\begin{figure}[t]
    \centering
    \includegraphics[width=0.98\linewidth]{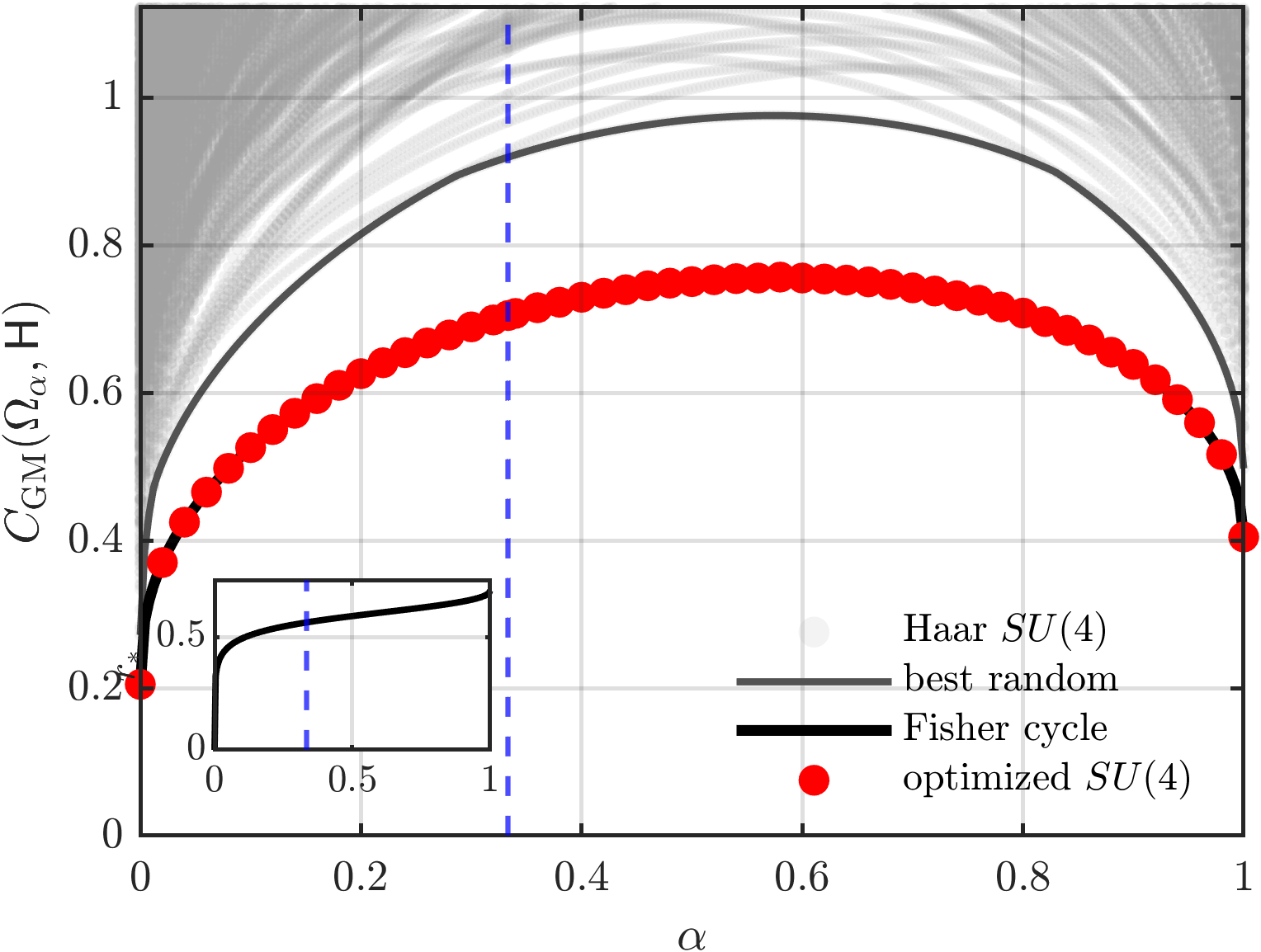}
    \caption{
    Single-qubit GM cost for the weighted three-delay task. Lower 
    values indicate superior estimation precision. The black solid curve 
    represents the optimized Fisher-cycle construction
    $U_r$. Gray points denote 
    random Haar $SU(4)$ gates, with the gray line tracking the empirical 
    random frontier at each $\alpha$. Red circles denote unrestricted local 
    optimizations over the full $SU(4)$ manifold. The
    vertical dashed line
    indicates the equal-weight loss baseline ($\alpha=1/3$), and the inset 
    displays the optimal parameter $r_\ast$ for the Fisher-cycle
    construction.}
    \label{fig:su4_frontier}
\end{figure}

The induced quantum channel factorizes the write-store step 
from the subsequent routing dynamics as 
$\Phi_{s,r}(\rho) = V_{\rm cyc} \Psi_{s,r}(\rho) V_{\rm cyc}^\dagger$. At the operating point $s_0=1/2$ and the maximally mixed reservoir state $\rho_\ast=\mathbb I_{\rm R}/2$, 
the channel and its derivative evaluate to
\begin{equation}
    \begin{aligned}
    \Psi_{s_0,r}(\rho)
    &=
    \frac{1+r}{2}\rho + \frac{1-r}{2} X\rho X, \\
    \left.
    \partial_s\Psi_{s,r}(\rho_\ast)
    \right|_{s_0}
    &=
    c_r Z,
    \qquad
    c_r=\frac{\pi}{2}\sqrt{1-r^2}.
    \end{aligned}
    \label{eq:single_channel}
\end{equation}
Consequently, the fresh perturbation is injected strictly along the $Z$ axis, 
while the write-store channel preserves the $X$ component and damps the $Y,Z$ sectors by $r$. Under the three-axis cyclic routing generated by $V_{\rm cyc}$, 
the resulting delayed responses follow the exact trajectory:
\begin{equation}
    \Delta_1 = c_r X, \qquad
    \Delta_2 = c_r Y, \qquad
    \Delta_3 = r c_r Z.
    \label{eq:single_Delta_route}
\end{equation}
Since these
responses occupy mutually orthogonal Bloch directions, their SLD inner
products vanish. This realization of orthogonal memory modes in quantum-state
tangent space parallels classical analyses based on Fisher information and
information-processing capacity
\cite{Ganguli2008MemoryTraces,Dambre2012InformationProcessing}. The resulting
QFIM is
\begin{equation}
    \mathsf{H}^{(3)}(r)
    =
    \pi^2(1-r^2)
    \operatorname{diag}(1,1,r^2),
    \label{eq:H3r}
\end{equation}
where the superscript denotes the task size, and $r$ explicitly governs the
trade-off between sensitivity to fresh inputs and memory retention
\cite{SupplementalMaterial}.

To test robustness away from equal delay weighting, we use
\(\Omega_\alpha=\operatorname{diag}[(1-\alpha)/2,(1-\alpha)/2,\alpha]\), with
\(0\le\alpha\le1\).  Increasing \(\alpha\) shifts the loss weight from the two
freshest delays to the oldest delay.  We benchmark this ansatz against general
two-qubit operations by sampling Haar-distributed \(SU(4)\) gates and by
performing unrestricted local optimizations over the full \(SU(4)\) manifold.
As shown in Fig.~\ref{fig:su4_frontier}, optimizing \(r\) at
each \(\alpha\) yields the same lower-cost frontier found by unrestricted local
optimizations over \(SU(4)\), while Haar-random gates populate higher costs.
We further establish that this frontier is globally optimal
\cite{SupplementalMaterial}.

\textit{Pauli-string cycles.---}
The single-qubit construction suggests how Fisher-orthogonal memory can
be extended to larger reservoirs. At a maximally mixed steady state, the QFIM
reduces to the Hilbert--Schmidt inner product up to a global constant, so
distinct Pauli strings are Fisher-orthogonal. We therefore generalize the
three-axis route by transporting stored responses along a CRO generated by a
Clifford unitary $V_K$,
\begin{equation}
    P_{j+1}=V_K P_j V_K^\dagger,
    \label{eq:Pauli_orbit}
\end{equation}
with $P_{K+1}=P_1$. A natural choice is a Clifford generator whose conjugation action realizes a
Singer cycle on the Pauli strings
\cite{Kern2010CyclicMUB}.  In this case \(V_K\) cyclically permutes
the \(2^N+1\) maximal commuting Pauli classes, which partition the
\(4^N-1\) nonidentity Pauli strings. Hence
an orbit that follows one representative string from each class yields a CRO of
length \(K=2^N+1\).  As a concrete illustration, a two-qubit reservoir
(\(N=2\)) uses a \(V_5\) generator to produce the cyclic five-step CRO shown in
Fig.~\ref{fig:k5_flow}.

To turn this route into recurrent dynamics, we replace the
single-qubit writing direction $Z$ and damping direction $X$ by anticommuting
reservoir Pauli strings $Q$ and $G$. The resulting write-store block is
\begin{equation}
    W_r(Q, G)
    =
    \begin{pmatrix}
    a_r\mathbb I_{\rm R} & b_r Q \\
    b_r G & -a_r GQ
    \end{pmatrix}_{\rm in},
    \label{eq:many_Wr}
\end{equation}
acting on the joint space of the input and the reservoir. Consequently, for the chosen $V_K$, the complete driven QRC step 
is obtained by combining this block with $V_K$ and the input rotation
$L_{\rm in}$,
\begin{equation}
    U_{r,K}
    =
    (\mathbb{I}_{\rm in}\otimes V_K)
    W_r(Q, G)
    (L_{\rm in}\otimes\mathbb I_{\rm R}) .
    \label{eq:many_Ur}
\end{equation}
Once $V_K$, $Q$, and $G$ are fixed, the Pauli algebra determines the
operator route and the hit locations, while $r$ sets the contraction at each
hit. The injection operator $Q$ fixes the entry point of the fresh perturbation at $P_1 = V_K Q V_K^\dagger$ (see Fig.~\ref{fig:k5_flow} for the $K=5$ example).

As $P_j$ propagates along the CRO, it acquires a factor $r$ whenever
it anticommutes with $G$ (a hit) and remains undamped otherwise. This sequence of hits directly defines the
binary damping pattern $\eta_K(Q,G)$ ($\eta_{K,j}=1$ for a hit and $0$
otherwise) visually encoded in Fig.~\ref{fig:k5_flow}.
Equivalently, up to an irrelevant overall sign,
\(\Delta_k=\pi\sqrt{1-r^2}\,
r^{n_{k-1}(\eta_K)}P_k/2^N\), where
\(n_{k-1}(\eta_K)=\sum_{j=1}^{k-1}\eta_{K,j}\) counts the prior hits
(\(n_0=0\)). Since the \(P_k\) are distinct Pauli strings, these response operators are
Fisher-orthogonal \cite{SupplementalMaterial}. The resulting \(K\)-delay QFIM is therefore strictly
diagonal:
\begin{equation}
    \left(\mathsf{H}^{(K)}_{\eta_K}(r)\right)_{kl}
    =
    \pi^2(1-r^2)r^{2n_{k-1}(\eta_K)}\delta_{kl}.
    \label{eq:H_eta}
\end{equation}

\begin{figure}[t]
    \centering
    \includegraphics[width=0.98\linewidth]{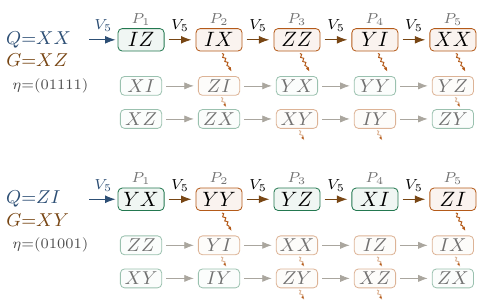}
    \caption{
   Two $K=5$ CROs generated by the same $V_5$. 
The blue arrow indicates the entry point $P_1 = V_5 Q V_5^\dagger$ of the fresh perturbation tangent operator. 
The damping string $G$ assigns the contractive hit pattern $\eta$. 
Pale rows explicitly illustrate the underlying parallel Singer partition structure, 
showing off-route Pauli strings transported simultaneously by the same generator, 
while orange boxes mark where they are damped.}
    \label{fig:k5_flow}
\end{figure}

Crucially, the Singer construction not only organizes the nonidentity Pauli
strings into complete routes, but also supplies the fading required for the
echo state property (ESP). At $s_0$, any nonidentity damping string $G$
anticommutes with at least one Pauli string on every nonidentity CRO. Thus, for
$0<r<1$, the linearized period map contracts every traceless mode,
establishing the ESP locally about $s_0$ with the maximally mixed state as the
unique fixed point. Numerical washout tests further support convergence under
finite input perturbations \cite{SupplementalMaterial}.

\begin{figure}[!t]
    \centering
    \includegraphics[width=0.98\linewidth]{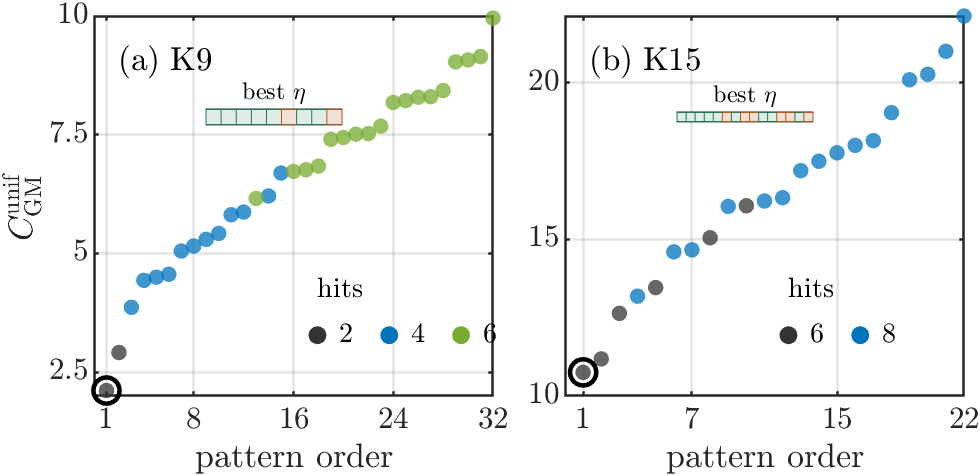}
    \caption{
    Fisher-cost landscape of Clifford damping patterns.
(a),(b) Optimized uniform-window cost \(C_{\rm GM}^{\rm unif}\) for the
distinct admissible hit words generated by the \(K=9\) Singer route
\(\mathcal C_9\) and the \(K=15\) product route \(\mathcal C_{15}\).
Points are ordered by cost and colored by the number of damping hits in one
period.  Open circles mark the best admissible words.  The inset in each
panel shows the corresponding best hit word \(\eta\), with orange cells
denoting \(\eta_j=1\) and green cells denoting \(\eta_j=0\).}
    \label{fig:pattern_scan}
\end{figure}

More flexible fading profiles can be obtained by combining smaller component
CROs into a composite route. Suppose the \(N\)-qubit reservoir is partitioned
into disjoint subsystems containing \(n_\mu\) qubits, with
\(\sum_\mu n_\mu=N\). Each subsystem is assigned a local Clifford generator
\(V_{K_\mu}\) that produces a CRO
\(\{P_a^{(\mu)}\}_{a=1}^{K_\mu}\) of period \(K_\mu\). At global step \(j\),
the \(\mu\)th subsystem occupies the phase
\([j]_\mu\equiv1+[(j-1)\bmod K_\mu]\). The composite generator
\(V_{\rm comp}=\bigotimes_\mu V_{K_\mu}\) therefore transports
\(P_j=\bigotimes_\mu P_{[j]_\mu}^{(\mu)}\), and the resulting route has the
least-common-multiple period
\(K_{\rm comp}=\operatorname{lcm}(K_1,K_2,\ldots)\). For pairwise-coprime
component periods, \(K_{\rm comp}=\prod_\mu K_\mu\), allowing the composite
route to exceed a single \(N\)-qubit Singer period without changing its
\(O(2^N)\) scaling \cite{SupplementalMaterial}.

For factorized operators
\(Q_{\rm comp}=\bigotimes_\mu Q_\mu\) and
\(G_{\rm comp}=\bigotimes_\mu G_\mu\), the global condition
\(\{Q_{\rm comp},G_{\rm comp}\}=0\) requires an odd number of local pairs
\((Q_\mu,G_\mu)\) to anticommute. For each subsystem, let
\(\eta^{(\mu)}\in\{0,1\}^{K_{\rm comp}}\) denote the hit word generated by
\((Q_\mu,G_\mu)\), periodically extended to the common period so that its
\(j\)th entry is evaluated at the local phase \([j]_\mu\). The composite
damping word is then the elementwise modulo-two sum
\begin{equation}
    \eta_{\rm comp}(Q_{\rm comp},G_{\rm comp})
    =
    \bigoplus_\mu \eta^{(\mu)}(Q_\mu,G_\mu) .
    \label{eq:eta_composite_xor}
\end{equation}
This pointwise parity rule can yield fading profiles unavailable to a
single Singer orbit. For pairwise-coprime Singer components, the product route
also has ESP whenever \(G_\mu\neq I\) on every subsystem
\cite{SupplementalMaterial}.

As an illustration of this framework, we examine a three-qubit system ($N=3$),
evaluating damping patterns from both the single Singer CRO ($K=9$) and the
product CRO generated by $V_{15}=V_5\otimes V_3$ ($K_{\rm comp}=15$). We find
32 and 22 distinct admissible hit words for the $K=9$ and
$K_{\rm comp}=15$ routes, respectively, all compatible with
$\{Q,G\}=0$ and local ESP. To
benchmark how these patterns distribute finite-window Fisher information, we
use the uniform-delay GM cost, optimizing \(r\) independently for each damping
word:
\begin{equation}
    C_{\rm GM}^{\rm unif}(\eta_K)
    =
    \min_{0<r<1}
    \frac{1}{K}
    \left[
        \sum_{k=1}^{K}
        \frac{1}{
        \sqrt{(\mathsf H_{\eta_K}(r))_{kk}}
        }
    \right]^2 .
    \label{eq:pattern_cost}
\end{equation}
In Figs.~\ref{fig:pattern_scan}(a) and~\ref{fig:pattern_scan}(b), we
sort the admissible patterns by their minimized costs for $K=9$ and $K=15$.
The cost generally increases with the number of hits and also depends on their
positions. The best words postpone their first hits, thereby preserving the
early-delay QFIs. However, an ESP-compatible word cannot contain more than
$2N-1$ consecutive zeros, which algebraically bounds this undamped prefix
\cite{SupplementalMaterial}.

\textit{Task benchmarks.---}
The Fisher analysis identifies the local response around a stationary input.
We test whether the resulting memory structure remains visible for
finite-amplitude signals, the full recurrent history, and an explicit readout protocol.  We drive three-qubit reservoirs with
\(s_t=s_0+\delta w\,u_t\), where \(u_t\in[-1,1]\) and \(\delta w=0.25\), and
consider the linear and product targets
\(y_t^{\rm lin}(q)=u_{t-q}\) and
\(y_t^{\rm prod}(q)=u_{t-q}u_{t-1}\) for \(q=2,\ldots,20\).

The comparison uses three fixed reservoirs. The two Clifford reservoirs use
the minimum-cost hit words from Fig.~\ref{fig:pattern_scan}, represented by
\((Q,G)=(XIX,ZXI)\) for \(K=9\) and \((XIX,YYX)\) for \(K=15\). For these
tasks, we adjust only the fading parameter, giving \(r=0.35\) and \(0.56\),
respectively. The baseline is a random Ising reservoir with
\(U_{\rm Ising}=e^{-\ii\tau_{\rm Ising}H_{\rm Ising}}\) and
\(H_{\rm Ising}=\sum_{i<j}J_{ij}X_iX_j+h^x\sum_iX_i+h^z\sum_iZ_i\)
\cite{Fujii2017Disordered,MartinezPena2023SpinIPC,Kubota2023Noise}.
It is selected on validation data from an ensemble satisfying the ESP, using
the aggregate linear and product errors over \(q=2,\ldots,8\), and is held
fixed for every task and delay \cite{SupplementalMaterial}.

With exact readout, the feature vector comprises the 63 nonidentity reservoir
Pauli expectations. With finite readout, the same features are reconstructed
from \(N_{\rm shot}=6400\) uniformly randomized local-Pauli measurements using
classical-shadow estimators \cite{Huang2020ClassicalShadows}. The three
reservoirs use the same trajectory lengths, input realizations, and readout
protocol. Only the classical linear readout is refitted for each target.

Figure~\ref{fig:task_benchmark} shows that the fixed \(K=9\) and \(K=15\)
reservoirs retain sharply resolved linear-memory windows, followed by a rapid
loss of accuracy near their respective route periods.  The Ising baseline
instead exhibits a gradual increase in error, and the Clifford thresholds
remain visible under finite-shot readout.  The product task probes the
second-order response and displays a nonmonotonic error profile. Its
main isolated peaks occur when the mixed response vanishes or aliases a
first-order memory direction, as fixed by the Clifford algebra
\cite{SupplementalMaterial}. Within the designed memory windows and away from
these route-specific obstructions, product information remains accessible to
the same randomized local-Pauli readout.

\begin{figure}[t]
    \centering
    \includegraphics[width=0.98\linewidth]{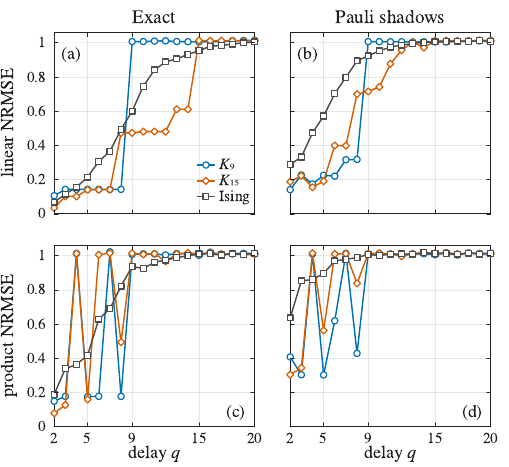}
    \caption{
Fixed-reservoir task benchmark. The top and bottom rows show NRMSE for the
linear target $u_{t-q}$ and product target $u_{t-q}u_{t-1}$, respectively.
The left column uses exact Pauli expectation values, while the right uses
local-Pauli shadows with $N_{\rm shot}=6400$. Blue circles, orange diamonds,
and gray squares denote the $K=9$, $K=15$, and Ising reservoirs. Shaded bands
indicate standard errors.
}
    \label{fig:task_benchmark}
\end{figure}

\textit{Discussion.---}
By casting QRC as a multiparameter estimation problem, we identify
Fisher orthogonality as a design principle for temporal readout under finite
measurements. The single-qubit cycle realizes this principle optimally, while
its Pauli-string extension provides solvable many-body reservoirs with
diagonal and programmable delay-space QFIMs. We further verify that
over the tested system sizes, the Singer routes retain better-conditioned
QFIMs than the sampled random Ising reservoirs and exhibit a
nonlinear-response window that grows linearly with $N$
\cite{SupplementalMaterial}. The routing generators are Clifford, whereas the
full driven dynamics is not because $W_r$ at generic $r$ and generic
finite-amplitude input states $\tau_s$ are non-Clifford. More broadly, our separation of writing, fading, and routing provides a
controlled starting point for future studies of whether and how quantum magic,
non-Markovianity, and temporal coherence can be harnessed in QRC.

\textit{Acknowledgments.---}We thank Zhe-Yu Shi, Chang Liu, Hongye Hu, Pengfei Zhang, Juan Yao and Liao Mao for insightful discussions. This work was supported by the Shanghai Oriental Elite Youth Project (Grant No. QNJY2025111).

\bibliographystyle{apsrev4-2}
\bibliography{refs}

\clearpage
\onecolumngrid

\setcounter{secnumdepth}{2}
\setcounter{section}{0}
\setcounter{subsection}{0}
\setcounter{equation}{0}
\setcounter{figure}{0}
\setcounter{table}{0}
\renewcommand{\thesection}{S\arabic{section}}
\renewcommand{\thesubsection}{\thesection.\arabic{subsection}}
\numberwithin{equation}{section}
\renewcommand{\thefigure}{S\arabic{figure}}
\renewcommand{\thetable}{S\arabic{table}}
\renewcommand{\theHsection}{S\arabic{section}}
\renewcommand{\theHsubsection}{S\arabic{section}.\arabic{subsection}}
\renewcommand{\theHequation}{S\arabic{section}.\arabic{equation}}
\renewcommand{\theHfigure}{S\arabic{figure}}
\renewcommand{\theHtable}{S\arabic{table}}
\renewcommand{\citenumfont}[1]{S#1}
\renewcommand{\bibnumfmt}[1]{[S#1]}

\begin{center}
{\bfseries Supplemental Material for\\
Fisher-Orthogonal Memory in Quantum Reservoir Computing}
\end{center}

\setcounter{tocdepth}{2}
\tableofcontents
\clearpage

This Supplemental Material is organized as follows.  Sections~S1 and~S2
formulate quantum reservoir computing (QRC) memory as a local
multiparameter-estimation problem and state
the corresponding precision bound.  Section~S3 proves the optimality of the
single-qubit Fisher cycle.  Section~S4 develops the many-body construction,
including its Clifford/Singer algebra, loss of initial-state dependence, and
second-order response.  Section~S5 specifies the numerical protocols used for
the figures in the Letter, while Sec.~S6 collects finite-amplitude, finite-copy, and
reservoir-size scaling tests.

\section{QRC Memory as a Local Multiparameter Estimation Problem}
\label{sec:S_local_qfi}

We formulate QRC memory as a local multiparameter estimation problem around
the stationary input \(s_0\).  Each classical input \(s\in[0,1]\) is encoded
into a fresh input qubit,
\begin{equation}
    \ket{\psi_s}
    =
    \cos\!\left(\frac{\pi s}{2}\right)\ket0
    +
    \sin\!\left(\frac{\pi s}{2}\right)\ket1,
    \tau_s=\proj{\psi_s}.
    \label{eq:S_input_encoding}
\end{equation}
For a fixed input-reservoir unitary \(U\), one driven reservoir step is the
channel
\begin{equation}
    \Phi_s(\rho)
    =
    \Tr_{\rm in}\!\left[
        U(\tau_s\otimes\rho)U^\dagger
    \right].
    \label{eq:S_driven_channel}
\end{equation}
Let \(s_0=1/2\), and let \(\rho_\ast\) be the fixed point of
\(\Phi_{s_0}\).  For a backward-ordered input history
\(\bm s=(s_1,\ldots,s_{L_{\rm hist}})\), where \(s_1\) is the most recent
input, the final
reservoir state is
\begin{equation}
    \rho_f(\bm s)
    =
    \Phi_{s_1}\circ\Phi_{s_2}\circ\cdots
    \circ\Phi_{s_{L_{\rm hist}}}(\rho_\ast).
    \label{eq:S_history_state}
\end{equation}
Expanding locally around the constant input history gives
\begin{equation}
    \delta\rho_f
    =
    \sum_{k=1}^{L_{\rm hist}}\Delta_k\,\delta s_k,
    \qquad
    \Delta_k
    =
    \Phi_{s_0}^{k-1}\!\left[\dot\Phi_{s_0}(\rho_\ast)\right],
    \label{eq:S_delta_k}
\end{equation}
where \(\dot\Phi_{s_0}=\partial_s\Phi_s|_{s=s_0}\).  The operators
\(\Delta_k\) are the tangent directions associated with different delay
coordinates.

For a state \(\rho\), define the symmetric-logarithmic-derivative (SLD)
Fisher metric
\begin{equation}
    g_\rho(O_1,O_2)
    =
    \operatorname{Re}\Tr\!\left[
        O_1\,\cJ_\rho^{-1}(O_2)
    \right],
    \qquad
    \cJ_\rho(X)=\frac{\rho X+X\rho}{2}.
    \label{eq:S_qfi_metric}
\end{equation}
For a rank-deficient state, \(\cJ_\rho^{-1}\) denotes the inverse on its
support, equivalently the Moore--Penrose inverse.  This convention is made
explicit by the spectral representation below.
For each delay tangent \(\Delta_k\), its SLD operator \(\mathcal L_k\) at the
operating state is defined by
\[
    \Delta_k
    =
    \cJ_{\rho_\ast}(\mathcal L_k)
    =
    \frac{\rho_\ast\mathcal L_k+\mathcal L_k\rho_\ast}{2},
    \qquad
    \mathcal L_k=\cJ_{\rho_\ast}^{-1}(\Delta_k).
\]
Thus \(\Delta_k\) is the change of the reservoir state caused by the \(k\)th
historical input, while \(\mathcal L_k\) is the corresponding SLD observable.
The delay-space quantum Fisher information matrix (QFIM) is then
\begin{equation}
    (\mathsf H)_{kl}
    =g_{\rho_\ast}(\Delta_k,\Delta_l)
    =\operatorname{Re}\Tr(\rho_\ast\mathcal L_k\mathcal L_l).
    \label{eq:S_H_def}
\end{equation}
Equivalently, if
\(\rho_\ast=\sum_a\lambda_a\ket{a}\bra{a}\), then
\begin{equation}
    (\mathsf H)_{kl}
    =
    2
    \sum_{a,b:\lambda_a+\lambda_b>0}
    \frac{
        \operatorname{Re}\!\left[
            \langle a|\Delta_k|b\rangle
            \langle b|\Delta_l|a\rangle
        \right]
    }{\lambda_a+\lambda_b}.
    \label{eq:S_spectral_H}
\end{equation}
Thus \(\mathsf H\) is a Gram matrix of delayed response operators in the local
SLD Fisher metric.  Its diagonal entries quantify the retained single-delay
quantum Fisher information (QFI), while its off-diagonal entries quantify Fisher overlap between different
delays.  If \(\bm c^T\mathsf H\bm c=0\) for some nonzero coefficient vector
\(\bm c\), then the corresponding temporal combination
\(\sum_k c_k\delta s_k\) is locally
invisible to any measurement.
\textbf{Accordingly, retaining a large response at each delay is not enough:
the QFIM must also keep different delays distinguishable from one another.}

At the maximally mixed state \(\rho_\ast=\id_{\rm R}/d\), where \(d=2^N\), the metric
reduces to
\begin{equation}
    g_{\id_{\rm R}/d}(O_1,O_2)=d\,\Tr(O_1O_2).
    \label{eq:S_maxmixed_metric}
\end{equation}
Since distinct nonidentity Pauli strings are Hilbert-Schmidt orthogonal, this
identity is the basis of the Pauli-cycle construction below.

\section{Measurement Precision and the Gill--Massar Bound}
\label{sec:S_gm_bound}

We now relate the QFIM defined in Sec.~\ref{sec:S_local_qfi} to measurement
precision.  For a smooth \(n_\theta\)-parameter family
\(\rho_{\bm\theta}\), the same construction replaces the delay tangent
\(\Delta_k\) by \(\partial_\mu\rho\) and defines
\(\mathcal L_\mu=\cJ_\rho^{-1}(\partial_\mu\rho)\).  Hence
\(\mathsf H_{\mu\nu}=g_\rho(\partial_\mu\rho,\partial_\nu\rho)
=\operatorname{Re}\Tr(\rho\mathcal L_\mu\mathcal L_\nu)\).
We work on the identifiable tangent subspace,
where \(\mathsf H\) is nonsingular; a singular direction with positive loss
weight has divergent estimation cost.  A positive-operator-valued measure
(POVM) \(\cB=\{E_x\}\) produces probabilities \(p_x=\Tr(\rho E_x)\) and the
classical Fisher information matrix (CFIM)
\begin{equation}
    \mathsf I_{\mu\nu}(\cB)
    =
    \sum_{x:p_x>0}
    \frac{\partial_\mu p_x\,\partial_\nu p_x}{p_x}.
    \label{eq:S_CFI}
\end{equation}
For one copy of a \(d_{\rm sys}\)-dimensional system, the GM trace inequality
states
\begin{equation}
    \Tr\!\left(\mathsf H^{-1}\mathsf I(\cB)\right)\le d_{\rm sys}-1.
    \label{eq:S_GM_trace}
\end{equation}
For \(N_{\rm shot}\) independent copies, the right-hand side is multiplied
by \(N_{\rm shot}\); equivalently,
\(\mathsf I_{\rm tot}=N_{\rm shot}\mathsf I(\cB)\) when the same measurement is applied to
each copy.  The POVM on an individual reservoir copy may be global across
all reservoir qubits.  The additive extension assumes separate, or more
generally separable, measurements across the copies; collective measurements
across different copies are not considered here.

To prove the bound, equip the Hermitian
operator space with
\(\langle O_1,O_2\rangle_\rho=\operatorname{Re}\Tr(\rho O_1O_2)\), and let
\(\mathcal T_{\rm SLD}=\operatorname{span}\{\mathcal L_\mu\}\).  The QFIM is
the Gram matrix of this generally nonorthogonal basis,
\(\mathsf H_{\mu\nu}=\langle\mathcal L_\mu,\mathcal L_\nu\rangle_\rho\).
For one POVM outcome \(E_x\), define
\(\bm\nabla p_x=(\partial_1p_x,\ldots,\partial_{n_\theta}p_x)^T\).  Its entries
are precisely the overlaps with the SLD directions,
\[
    \partial_\mu p_x
    =\Tr[(\partial_\mu\rho)E_x]
    =\langle E_x,\mathcal L_\mu\rangle_\rho.
\]
Writing the orthogonal projection as
\(\Pi_{\mathcal T}E_x=\sum_\mu c_\mu\mathcal L_\mu\), the projection
conditions give
\(\bm\nabla p_x=\mathsf H\bm c\), and hence
\[
    \bm c=\mathsf H^{-1}\bm\nabla p_x,
    \qquad
    \|\Pi_{\mathcal T}E_x\|_\rho^2
    =(\bm\nabla p_x)^T\mathsf H^{-1}(\bm\nabla p_x).
\]
The identity direction carries only probability normalization and is
orthogonal to every SLD because
\(\langle\id,\mathcal L_\mu\rangle_\rho
=\Tr(\partial_\mu\rho)=0\).  The identity component of \(E_x\) is
\(p_x\id\), since \(\langle\id,E_x\rangle_\rho=p_x\) and
\(\|\id\|_\rho^2=1\).  Therefore the SLD projection cannot be longer than
the full component remaining after normalization is removed,
\(E_x-p_x\id\).  This gives
\begin{equation}
    (\bm\nabla p_x)^T\mathsf H^{-1}(\bm\nabla p_x)
    \le
    \|E_x-p_x\id\|_\rho^2
    =
    \Tr(\rho E_x^2)-p_x^2.
    \label{eq:S_projection_bound}
\end{equation}
Since \(0\le E_x\le\id\), we have
\(E_x^2\le\lambda_{\max}(E_x)E_x\le \Tr(E_x)E_x\), and hence
\(\Tr(\rho E_x^2)\le \Tr(E_x)p_x\).  Dividing
Eq.~\eqref{eq:S_projection_bound} by \(p_x\) and summing over the outcomes
gives
\begin{equation}
    \Tr\!\left(\mathsf H^{-1}\mathsf I(\cB)\right)
    =
    \sum_x
    \frac{(\bm\nabla p_x)^T\mathsf H^{-1}(\bm\nabla p_x)}{p_x}
    \le
    \sum_x\left[\frac{\Tr(\rho E_x^2)}{p_x}-p_x\right]
    \le
    \sum_x\left[\Tr(E_x)-p_x\right]
    =
    d_{\rm sys}-1.
    \label{eq:S_GM_proof}
\end{equation}
For each outcome, the quantity
\((\bm\nabla p_x)^T\mathsf H^{-1}(\bm\nabla p_x)/p_x\) measures the SLD-tangent
component resolved by that outcome, normalized by its probability.
Equation~\eqref{eq:S_GM_proof} bounds the sum of these resolved components by
\(d_{\rm sys}-1\) per copy.  A POVM may have arbitrarily many outcomes, but
additional outcomes only divide this fixed total among more probabilities;
they do not create new tangent information.

For a positive loss matrix \(\Omega\), the locally unbiased mean-square error
satisfies
\begin{equation}
    \Tr\!\left[\Omega\,\operatorname{Cov}(\hat{\bm\theta})\right]
    \ge
    \Tr(\Omega \mathsf I_{\rm tot}^{-1}).
    \label{eq:S_weighted_CR}
\end{equation}
Every realizable positive one-copy CFIM belongs to the larger set
\begin{equation}
    \left\{\mathsf I>0:
    \Tr(\mathsf H^{-1}\mathsf I)\le d_{\rm sys}-1\right\}.
    \label{eq:S_GM_relaxed_set}
\end{equation}
Minimizing the right-hand side of Eq.~\eqref{eq:S_weighted_CR} over this
enlarged set can only lower the optimum and therefore gives a
measurement-independent lower bound.  Define
\[
    \mathsf J
    =\mathsf H^{-1/2}\mathsf I\mathsf H^{-1/2},
    \qquad
    \mathsf A
    =\mathsf H^{-1/2}\Omega\mathsf H^{-1/2},
\]
the constraint becomes \(\Tr\mathsf J\le d_{\rm sys}-1\), while the one-copy weighted
cost becomes
\[
    \Tr(\Omega\mathsf I^{-1})
    =\Tr(\mathsf A\mathsf J^{-1}).
\]
The optimum uses the full budget, \(\Tr\mathsf J=d_{\rm sys}-1\).  The matrix
Cauchy--Schwarz inequality gives
\[
    \Tr(\mathsf A\mathsf J^{-1})\,\Tr\mathsf J
    \ge
    \bigl(\Tr\sqrt{\mathsf A}\bigr)^2,
\]
with equality for
\(\mathsf J_\ast=(d_{\rm sys}-1)\sqrt{\mathsf A}/\Tr\sqrt{\mathsf A}\).
The resulting relaxed optimum is
\begin{equation}
    C_{\rm GM}
    =
    \frac{
        \left\{
        \Tr\!\left[
            \left(
            \mathsf H^{-1/2}\Omega \mathsf H^{-1/2}
            \right)^{1/2}
        \right]
        \right\}^2
    }{d_{\rm sys}-1}.
    \label{eq:S_GM_cost_full}
\end{equation}
Thus, for \(N_{\rm shot}\) copies, the corresponding weighted estimation
error obeys
\begin{equation}
    \Tr\!\left[
        \Omega\,\operatorname{Cov}(\hat{\bm\theta})
    \right]
    \ge
    \frac{C_{\rm GM}}{N_{\rm shot}} .
    \label{eq:S_GM_shot_scaling}
\end{equation}
Throughout this work, \(C_{\rm GM}\) denotes this shot-independent error
coefficient rather than the finite-\(N_{\rm shot}\) error itself.
In the main text, the common factor \(1/(d_{\rm sys}-1)\) is omitted when comparing
patterns with the same reservoir dimension.  If \(\mathsf H\) and \(\Omega\) are
diagonal, Eq.~\eqref{eq:S_GM_cost_full} becomes
\begin{equation}
    C_{\rm GM}
    =
    \frac{
        \left(
            \sum_{\mu=1}^{n_\theta}
            \sqrt{\Omega_{\mu\mu}/\mathsf H_{\mu\mu}}
        \right)^2
    }{d_{\rm sys}-1}.
    \label{eq:S_GM_diagonal}
\end{equation}

For measurements performed separately on each reservoir copy,
\(C_{\rm GM}\) depends only on \(\mathsf H\) and \(\Omega\), rather than on a
particular POVM.  It is the coefficient of the GM lower bound, not the error
attained by an arbitrary measurement.  The finite-shot performance of the
local-Pauli readout is evaluated directly from its joint outcome distribution
in Secs.~\ref{sec:S_numerical_protocols}
and~\ref{sec:S_robustness_scaling}.

\section{Optimal Single-Qubit Fisher Cycle}
\label{sec:S_single_qubit}

\subsection{Explicit cycle, channel, and three-delay QFIM}
\label{sec:S_explicit_single_cycle}

The one-qubit reservoir construction in the main text is
\begin{equation}
    \begin{aligned}
    U_r
    &=
    (\id_{\rm in}\otimes V_{\rm cyc})\,
    W_r\,
    (L_{\rm in}\otimes\id_{\rm R}),\\
    L_{\rm in}&=e^{+\ii\pi Y/4},
    \qquad
    V_{\rm cyc}=H_{\rm Had}S^\dagger,\\
    W_r
    &=
    \begin{pmatrix}
    a_r\id_{\rm R} & b_r Z\\
    b_r X & \ii a_r Y
    \end{pmatrix}_{\rm in},
    \end{aligned}
    \label{eq:S_single_U}
\end{equation}
with \(a_r=\sqrt{(1+r)/2}\), \(b_r=\sqrt{(1-r)/2}\), and
\(0<r<1\).  The block structure is in the input-qubit basis, while the Pauli
operators act on the reservoir qubit.  The reservoir Clifford
\(V_{\rm cyc}\) cyclically routes \(Z\mapsto X\mapsto Y\mapsto Z\).

The final factor \(V_{\rm cyc}\) acts after the input qubit is discarded, so
the channel factorizes as
\(\Phi_{s,r}(\rho)=V_{\rm cyc}\Psi_{s,r}(\rho)V_{\rm cyc}^\dagger\).  To derive
\(\Psi_{s,r}\), first rotate the encoded input:
\begin{equation}
    L_{\rm in}\ket{\psi_s}
    =
    \ell_s\ket0+m_s\ket1,
    \qquad
    \ell_s=\frac{\cos(\pi s/2)+\sin(\pi s/2)}{\sqrt2},
    \quad
    m_s=\frac{\sin(\pi s/2)-\cos(\pi s/2)}{\sqrt2}.
    \label{eq:S_rotated_input}
\end{equation}
At the operating point \(s_0=1/2\),
\(\ell_{s_0}=1\), \(m_{s_0}=0\),
\(\dot\ell_{s_0}=0\), and \(\dot m_{s_0}=\pi/2\), where the dot denotes
\(\partial_s\).  The write-store block then gives the two Kraus operators
\begin{equation}
    K_0(s)=\ell_s a_r\id_{\rm R}+m_s b_r Z,
    \qquad
    K_1(s)=\ell_s b_r X+m_s\ii a_rY,
    \label{eq:S_single_kraus}
\end{equation}
so that \(\Psi_{s,r}(\rho)=\sum_{\nu=0,1}K_\nu(s)\rho K_\nu^\dagger(s)\).
Thus \(K_0(s_0)=a_r\id_{\rm R}\), \(K_1(s_0)=b_rX\), while
\(\dot K_0(s_0)=(\pi/2)b_rZ\) and
\(\dot K_1(s_0)=(\pi/2)\ii a_rY\).  Substituting the zeroth-order Kraus
operators gives
\begin{equation}
    \Psi_{s_0,r}(\rho)
    =
    \frac{1+r}{2}\rho+\frac{1-r}{2}X\rho X,
    \label{eq:S_single_channel}
\end{equation}
while differentiating the Kraus representation gives
\begin{equation}
    \begin{aligned}
    \dot\Psi_{s_0,r}(\rho)
    &=
    \sum_{\nu=0,1}
    \left[
        \dot K_\nu(s_0)\rho K_\nu^\dagger(s_0)
        +
        K_\nu(s_0)\rho \dot K_\nu^\dagger(s_0)
    \right] \\
    &=
    \frac{\pi a_rb_r}{2}
    \left[
        Z\rho+\rho Z
        +
        \ii\left(Y\rho X-X\rho Y\right)
    \right].
    \end{aligned}
    \label{eq:S_single_dotPsi}
\end{equation}
At \(\rho_\ast=\id_{\rm R}/2\), the two brackets in
Eq.~\eqref{eq:S_single_dotPsi} contribute equally:
\(Z\rho_\ast+\rho_\ast Z=Z\) and
\(\ii(Y\rho_\ast X-X\rho_\ast Y)=Z\).  Therefore
\begin{equation}
    \left.\partial_s\Psi_{s,r}(\rho_\ast)\right|_{s_0}
    =
    c_r Z,
    \qquad
    c_r=\frac{\pi}{2}\sqrt{1-r^2}.
    \label{eq:S_single_injection}
\end{equation}
The same operating-point channel acts on Pauli components as
\(\Psi_{s_0,r}(X)=X\), \(\Psi_{s_0,r}(Y)=rY\), and
\(\Psi_{s_0,r}(Z)=rZ\).
After the final cyclic routing, the three most recent delayed tangents are
\begin{equation}
    \Delta_1=c_rX,\qquad
    \Delta_2=c_rY,\qquad
    \Delta_3=rc_rZ.
    \label{eq:S_single_Deltas}
\end{equation}
Since \(g_{\id_{\rm R}/2}(O_1,O_2)=2\Tr(O_1O_2)\),
Eq.~\eqref{eq:S_single_Deltas} gives
\begin{equation}
    \mathsf H^{(3)}(r)
    =
    \pi^2(1-r^2)\operatorname{diag}(1,1,r^2).
    \label{eq:S_single_H3}
\end{equation}

\subsection{Global optimality of the single-qubit cycle}

The proof rests on three facts.  First, injecting a pure qubit gives a
rank-two reservoir channel whose contraction spectrum is unchanged by Fisher
normalization.  Second, QFI contractivity imposes a common budget on fresh
writing and delayed retention.  Third, a chronological QR decomposition
measures how much independent Fisher information each additional delay
contributes.  Once these facts are separated, the GM bound follows from a
single norm inequality.

We begin by fixing the tangent-space coordinates.  Write a qubit state as
\(\rho=(\id+\bm x\cdot\bm\sigma)/2\).  A Bloch tangent \(\delta\bm x\),
corresponding to
\(\delta\rho=(\delta\bm x\cdot\bm\sigma)/2\), has squared SLD Fisher length
\[
    g_\rho(\delta\rho,\delta\rho)
    =
    \delta\bm x^T\mathcal G(\bm x)\delta\bm x,
\]
where \(\mathcal G(\bm x)\) is the positive metric matrix below.  It is
not sufficient to compare tangents through their ordinary Bloch-vector
lengths, because the local distinguishability of a perturbation depends on
both the operating state and its direction.  To compare the information
carried by tangents before and after a channel, we therefore use the
Fisher-normalized coordinate
\[
    \bm\xi=\mathcal G(\bm x)^{1/2}\delta\bm x.
\]
\textbf{The coordinate \(\bm\xi\) introduces no new physical degree of
freedom: it is the same state perturbation measured in units of local quantum
distinguishability.}  By construction,
\[
    g_\rho(\delta\rho,\delta\rho)=\|\bm\xi\|^2,
\]
so its squared Euclidean norm is exactly the SLD QFI carried by the
perturbation.  In these coordinates, the action of a channel directly
quantifies how much of that information is retained.

A qubit channel
\(\bm x_{\rm out}=A\bm x_{\rm in}+\bm t\) propagates infinitesimal tangents as
\(\delta\bm x_{\rm out}=A\delta\bm x_{\rm in}\); the affine shift \(\bm t\)
does not enter this derivative.  Because the input and output states generally
carry different local metrics, their normalized tangent coordinates obey
\[
    \bm\xi_{\rm out}
    =
    \mathcal G(\bm x_{\rm out})^{1/2}\delta\bm x_{\rm out}
    =
    \mathcal G(\bm x_{\rm out})^{1/2}A\delta\bm x_{\rm in}
    =
    \mathcal G(\bm x_{\rm out})^{1/2}A
    \mathcal G(\bm x_{\rm in})^{-1/2}\bm\xi_{\rm in}.
\]
This motivates the Fisher-normalized tangent map
\begin{equation}
    \mathcal G(\bm x)
    =
    \id_3+\frac{\bm x\bm x^T}{1-\|\bm x\|^2}
    =
    (\id_3-\bm x\bm x^T)^{-1},
    \qquad
    \widehat A_{\bm x_{\rm in}\to\bm x_{\rm out}}
    =
    \mathcal G(\bm x_{\rm out})^{1/2}
    A
    \mathcal G(\bm x_{\rm in})^{-1/2}.
    \label{eq:S_fixed_point_metric}
\end{equation}
Thus \(\widehat A\) is not a new physical channel.  It is the derivative of
the same channel expressed in coordinates in which QFI length is Euclidean:
the right factor converts a normalized input vector into a Bloch tangent,
\(A\) propagates that tangent, and the left factor normalizes the output in
its own Fisher metric.  Consequently,
\[
    \frac{
        g_{\rho_{\rm out}}(\delta\rho_{\rm out},\delta\rho_{\rm out})
    }{
        g_{\rho_{\rm in}}(\delta\rho_{\rm in},\delta\rho_{\rm in})
    }
    =
    \frac{\|\widehat A\bm\xi_{\rm in}\|^2}{\|\bm\xi_{\rm in}\|^2}.
\]
The singular values of \(\widehat A\) are therefore the principal
single-step retention factors for QFI amplitudes, rather than for ordinary
Bloch-vector lengths.  Their pairwise products control the area of the
parallelogram spanned by two Fisher-normalized tangents and hence their
simultaneous two-direction retention.  At a physical
fixed point, \(\bm x_{\rm in}=\bm x_{\rm out}=\bm x_\ast\), so the map reduces
to
\[
    \widehat A_\ast
    =
    \mathcal G_\ast^{1/2}A\mathcal G_\ast^{-1/2},
    \qquad
    \mathcal G_\ast=\mathcal G(\bm x_\ast).
\]

\begin{lemma}[Rank-two Fisher spectrum]
For a qubit channel induced by a pure input qubit, let \(\bm x_\ast\) be a
full-rank fixed point, and let \(\operatorname{sv}\) denote the ordered
singular-value spectrum.  The Fisher-normalized map satisfies
\begin{equation}
    \operatorname{sv}(\widehat A_\ast)
    =
    \operatorname{sv}(A)
    =
    (\nu_1,\nu_2,\kappa),
    \qquad
    \kappa=\nu_1\nu_2=\sqrt{|\det A|},
    \label{eq:S_metric_contraction_identity}
\end{equation}
with \(1\ge\nu_1\ge\nu_2\ge\kappa\ge0\).  In particular,
\(\sigma_{\min}(\widehat A_\ast)=\kappa\).  Moreover, for any two tangent
coordinates \(\bm u,\bm v\),
\[
    \|\widehat A_\ast\bm u\times\widehat A_\ast\bm v\|
    \le
    \kappa\|\bm u\times\bm v\|.
\]
Thus \(\kappa\) is also the largest possible one-step retention factor for
the area spanned by two Fisher-normalized tangents.
\end{lemma}

\begin{proof}
Choose the input basis at the operating point so that
\(\ket{\psi_{s_0}}=\ket0\).  Tracing out this pure input leaves at most two
Kraus operators.  Up to a Kraus-space rotation and unitary transformations
at the channel endpoints, every such qubit channel admits the canonical form
\cite{SMRuskai2002QubitChannels}
\begin{equation}
    F_0
    =
    \begin{pmatrix}
        \cos\gamma_2&0\\
        0&\cos\gamma_1
    \end{pmatrix},
    \qquad
    F_1
    =
    \begin{pmatrix}
        0&\sin\gamma_1\\
        \sin\gamma_2&0
    \end{pmatrix},
    \label{eq:S_rank_two_normal_form}
\end{equation}
where axis permutations and paired sign changes allow
\(0\le\gamma_2\le\gamma_1\) and
\(\gamma_1+\gamma_2\le\pi/2\).  Direct substitution
gives
\begin{equation}
    \begin{gathered}
    \bm x\longmapsto D\bm x+\bm t_D,\qquad
    D=\operatorname{diag}(\nu_1,\nu_2,\kappa),\qquad
    \bm t_D=t_3\bm e_Z,\\
    \nu_1=\cos(\gamma_1-\gamma_2),\quad
    \nu_2=\cos(\gamma_1+\gamma_2),\quad
    \kappa=\nu_1\nu_2,\quad
    t_3^2=(1-\nu_1^2)(1-\nu_2^2),\quad
    1\ge\nu_1\ge\nu_2\ge0.
    \end{gathered}
    \label{eq:S_rank_two_parameters}
\end{equation}
Here \(\bm e_X,\bm e_Y,\bm e_Z\) denote the Cartesian unit vectors of Bloch
space.  Endpoint rotations restore the full affine channel as
\(A=R_{\rm out}DR_{\rm in}\) and
\(\bm t=R_{\rm out}\bm t_D\), without changing
\(\kappa=\sqrt{|\det A|}=\nu_1\nu_2\), where
\(R_{\rm in},R_{\rm out}\in SO(3)\).

At the fixed point, the affine equation in the physical Bloch frame becomes
\begin{equation}
    \bm x_\ast
    =R_{\rm out}\!\left(DR_{\rm in}\bm x_\ast+\bm t_D\right)
    \quad\Longleftrightarrow\quad
    \underbrace{R_{\rm out}^T\bm x_\ast}_{\bm x_L}
    =D\underbrace{R_{\rm in}\bm x_\ast}_{\bm x_R}+\bm t_D .
    \label{eq:S_canonical_fixed_coordinates}
\end{equation}
Thus \(\bm x_R\) and \(\bm x_L\) are coordinates of the same physical fixed
point in the input and output canonical frames, and
\(\|\bm x_R\|=\|\bm x_L\|=\|\bm x_\ast\|\).  Substituting
\(A=R_{\rm out}DR_{\rm in}\) into Eq.~\eqref{eq:S_fixed_point_metric} and
using
\(\mathcal G(R\bm x)^{\pm1/2}=R\mathcal G(\bm x)^{\pm1/2}R^T\) gives
\[
    \begin{aligned}
    \widehat A_\ast
    &=\mathcal G(\bm x_\ast)^{1/2}
      R_{\rm out}DR_{\rm in}
      \mathcal G(\bm x_\ast)^{-1/2}\\
    &=R_{\rm out}
      \mathcal G(\bm x_L)^{1/2}D
      \mathcal G(\bm x_R)^{-1/2}
      R_{\rm in}
      \equiv R_{\rm out}\widehat D_\ast R_{\rm in}.
    \end{aligned}
\]
Thus \(\widehat A_\ast\) and \(\widehat D_\ast\) have the same singular
values.

The squared singular values of \(\widehat D_\ast\) are the eigenvalues of
\(\widehat D_\ast^T\widehat D_\ast\).  For compactness, write
\(\mathcal G_R=\mathcal G(\bm x_R)\) and
\(\mathcal G_L=\mathcal G(\bm x_L)\).  The metric square roots can be
removed before taking the determinant because
\[
    \widehat D_\ast^T\widehat D_\ast-\lambda\id_3
    =
    \mathcal G_R^{-1/2}
    \bigl(D\mathcal G_LD-\lambda\mathcal G_R\bigr)
    \mathcal G_R^{-1/2}.
\]
Substituting
\(\bm x_L=D\bm x_R+t_3\bm e_Z\), the explicit metric in
Eq.~\eqref{eq:S_fixed_point_metric}, and the rank-two identities in
Eq.~\eqref{eq:S_rank_two_parameters} reduces the determinant to a polynomial
in \(\lambda\).  The fixed-point relation
\(\|\bm x_L\|=\|\bm x_R\|\) cancels the remaining dependence on the
fixed-point coordinates, leaving
\begin{equation}
    \begin{aligned}
    \det(\widehat D_\ast^T\widehat D_\ast-\lambda\id_3)
    &=
    \frac{
    \det\!\left(D\mathcal G_LD-\lambda\mathcal G_R\right)
    }{
    \det\mathcal G_R
    }
    &=
    (\nu_1^2-\lambda)
    (\nu_2^2-\lambda)
    (\kappa^2-\lambda).
    \end{aligned}
    \label{eq:S_metric_characteristic_factorization}
\end{equation}
Hence
\(\operatorname{sv}(\widehat D_\ast)=(\nu_1,\nu_2,\kappa)\).  The
endpoint rotations then give the same spectrum for \(\widehat A_\ast\), while
\(A=R_{\rm out}DR_{\rm in}\) gives it for \(A\).  Notice that
\(\widehat A_\ast=\mathcal G_\ast^{1/2}A\mathcal G_\ast^{-1/2}\) is similar
to \(A\), so equality of their eigenvalues would hold for any qubit channel.
\textbf{The nontrivial result here is the equality of their singular values,
which follows from the Kraus-rank-two identities in
Eq.~\eqref{eq:S_rank_two_parameters} and need not hold at higher Kraus rank.}
Finally, a linear map with
singular values \(s_1\ge s_2\ge s_3\) can enlarge any two-dimensional area
by at most \(s_1s_2\).  Here this factor is
\(\nu_1\nu_2=\kappa\), proving the area bound and the lemma.
\end{proof}

The spectrum lemma quantifies how an existing reservoir tangent survives one
driven step.  Memory performance also depends on how strongly that same step
writes the newly arriving input.  \textbf{Fresh writing and delayed retention
are not independent resources: the same input--reservoir unitary must
accommodate both within a common QFI budget.}  To express this tradeoff, vary
only the fresh input while keeping the incoming reservoir at \(\rho_\ast\),
and define the Fisher-normalized write vector \(\bm w\) by
\begin{equation}
    \dot\Phi_{s_0}(\rho_\ast)
    =
    \frac{\bm g_{\rm w}\cdot\bm\sigma}{2},
    \qquad
    \bm w=\mathcal G_\ast^{1/2}\bm g_{\rm w},
    \qquad
    h_{\rm w}=\|\bm w\|^2,
    \qquad
    \bm n_{\rm w}=\bm w/\sqrt{h_{\rm w}}.
    \label{eq:S_fresh_write_QFI}
\end{equation}
Here \(\bm\sigma=(X,Y,Z)\) is the Pauli vector.

The common budget follows from the monotonicity of QFI under a
parameter-independent quantum channel.  Any POVM \(\{E_x\}\) performed after
a channel \(\Lambda\)
can be pulled back to the POVM \(\{\Lambda^\dagger(E_x)\}\) before the
channel, with exactly the same outcome probabilities.  \textbf{Every
measurement available after the channel is therefore already represented
before it, so the QFI cannot increase.}  Applying this principle jointly to
the fresh input tangent and the stored reservoir tangent gives the following
writing--retention budget.

\begin{lemma}[Writing--retention budget]
For the pure input encoding used here, whose QFI is \(\pi^2\), the fresh
write strength obeys
\begin{equation}
    h_{\rm w}\le\pi^2(1-\kappa^2).
    \label{eq:S_single_write_bound}
\end{equation}
Equality requires \(\bm n_{\rm w}\) to be a left singular direction of
\(\widehat A_\ast\) with singular value \(\kappa\), together with saturation
of the joint QFI contraction along that direction.
\end{lemma}

\begin{proof}
Let \(\delta\bm x\) be an incoming reservoir tangent carrying earlier input
information.  \textbf{Contractivity must hold for every joint direction
\((\delta s,\delta\bm x)\), not only for the fresh-input direction
\((\delta s,0)\) and the stored-reservoir direction
\((0,\delta\bm x)\) separately.}  This includes arbitrary linear
combinations of fresh and stored information.  To first order, one driven
step maps such a joint tangent to
\(\delta\bm y=\bm g_{\rm w}\delta s+A\delta\bm x\).  Introduce the
Fisher-normalized four-component input coordinate and its corresponding
output map,
\begin{equation}
\begin{gathered}
    \bm z\equiv
    \binom{\pi\delta s}{\mathcal G_\ast^{1/2}\delta\bm x},
    \qquad
    \mathsf T\equiv
    \begin{pmatrix}\bm w/\pi&\widehat A_\ast\end{pmatrix},\\[-2pt]
    \mathcal G_\ast^{1/2}\delta\bm y
    =\mathsf T\bm z
    =\frac{\bm w}{\pi}(\pi\delta s)
    +\widehat A_\ast\mathcal G_\ast^{1/2}\delta\bm x .
\end{gathered}
\label{eq:S_output_QFI_budget}
\end{equation}
The input and output QFIs are therefore \(\|\bm z\|^2\) and
\(\|\mathsf T\bm z\|^2\), respectively.  Contractivity for every
\(\bm z\) is equivalent to an operator-norm bound, which may be written on
the three-dimensional output space as
\begin{equation}
\begin{aligned}
    \|\mathsf T\bm z\|^2\le\|\bm z\|^2
    \quad\text{for all }\bm z
    &\Longleftrightarrow
    \|\mathsf T\|_{\rm op}\le1
    \Longleftrightarrow
    \mathsf T\mathsf T^T\preceq\id_3,\\
    \mathsf T\mathsf T^T
    &=
    \frac{\bm w\bm w^T}{\pi^2}
    +\widehat A_\ast\widehat A_\ast^T
    \preceq\id_3 .
\end{aligned}
\label{eq:S_output_QFI_matrix_budget}
\end{equation}
Here \(\preceq\) denotes the positive-semidefinite matrix order.  Although
\(\|\mathsf T\bm z\|^2\) contains a cross term between the fresh and stored
tangents, requiring the inequality for every \(\bm z\) is exactly equivalent
to the operator bound above; no cross term has been discarded.  Equivalently,
no linear combination of fresh and stored tangent information can gain QFI
through the joint channel.
Projecting its matrix inequality onto \(\bm n_{\rm w}\) gives
\begin{equation}
    \frac{h_{\rm w}}{\pi^2}
    +
    \|\widehat A_\ast^T\bm n_{\rm w}\|^2
    \le1.
\end{equation}
The spectrum lemma implies
\(\|\widehat A_\ast^T\bm n_{\rm w}\|\ge\kappa\), proving
Eq.~\eqref{eq:S_single_write_bound} and its equality condition.
\end{proof}

The three normalized delayed directions are
\(\bm d_k=\widehat A_\ast^{\,k-1}\bm n_{\rm w}\).  Let
\(\mathsf D_j=(\bm d_1,\ldots,\bm d_j)\), so the three-delay QFIM is
\(\mathsf H_U=h_{\rm w}\mathsf D_3^T\mathsf D_3\).  Its chronological QR
decomposition gives
\begin{equation}
    \sqrt{h_{\rm w}}\,\mathsf D_3
    =
    \mathsf O_{\rm QR}\mathsf C,
    \qquad
    \mathsf O_{\rm QR}^T\mathsf O_{\rm QR}=\id_3,
    \qquad
    \mathsf H_U=\mathsf C^T\mathsf C .
    \label{eq:S_causal_QR}
\end{equation}
We choose the standard QR convention in which \(\mathsf C\) is upper
triangular and \(\mathsf C_{kk}>0\).  The columns of
\(\mathsf O_{\rm QR}\) are the orthonormalized delay directions; this symbol
is distinct from the Pauli write operator \(Q\).

\begin{lemma}[Independent Fisher information from successive delays]
The diagonal entries of the chronological QR factor satisfy
\begin{equation}
    \mathsf C_{11}=\sqrt{h_{\rm w}},
    \qquad
    \mathsf C_{22}\le\nu_1\mathsf C_{11},
    \qquad
    \mathsf C_{33}\le\kappa\mathsf C_{11}.
    \label{eq:S_causal_Fisher_bounds}
\end{equation}
Moreover, \textbf{\(\mathsf C_{kk}^2\) is the independent Fisher information
added by the \(k\)th delay after the preceding delayed directions have been
projected out.}
\end{lemma}

\begin{proof}
Because \(\bm d_1=\bm n_{\rm w}\) is normalized,
\(\mathsf C_{11}=\sqrt{h_{\rm w}}\).  The next QR diagonal entry is the
component of \(\bm d_2\) orthogonal to \(\bm d_1\).  Since \(\bm d_1\) is a
unit vector, this component is also the area of the parallelogram spanned by
\(\bm d_1\) and \(\bm d_2\):
\[
    \frac{\mathsf C_{22}}{\sqrt{h_{\rm w}}}
    =\|(\id_3-\bm d_1\bm d_1^T)\bm d_2\|
    =\|\bm d_1\times\bm d_2\|
    \le\|\bm d_2\|
    \le\nu_1.
\]
For the third direction, the QR diagonal entry is the height of \(\bm d_3\)
above the plane spanned by \(\bm d_1\) and \(\bm d_2\).  It equals the
corresponding volume divided by the base area:
\begin{equation}
    \begin{aligned}
    \frac{\mathsf C_{33}}{\sqrt{h_{\rm w}}}
    &=
    \frac{|\bm d_1\cdot(\bm d_2\times\bm d_3)|}
    {\|\bm d_1\times\bm d_2\|}\\
    &\le
    \frac{\|\bm d_2\times\bm d_3\|}
    {\|\bm d_1\times\bm d_2\|}
    =
    \frac{\|\widehat A_\ast\bm d_1
    \times\widehat A_\ast\bm d_2\|}
    {\|\bm d_1\times\bm d_2\|}\\
    &\le
    \kappa,
    \end{aligned}
\end{equation}
where \(\bm d_2=\widehat A_\ast\bm d_1\),
\(\bm d_3=\widehat A_\ast\bm d_2\), and the last inequality is the Fisher-area
bound proved above.  This proves Eq.~\eqref{eq:S_causal_Fisher_bounds}.

For the information interpretation, let \(\Pi_{k-1}\) project onto
\(\operatorname{span}(\bm d_1,\ldots,\bm d_{k-1})\), with \(\Pi_0=0\).
Then
\[
    \mathsf C_{kk}^2
    =h_{\rm w}\|(\id_3-\Pi_{k-1})\bm d_k\|^2
    =\frac{\det\mathsf H_U^{(k)}}{\det\mathsf H_U^{(k-1)}},
\]
where \(\mathsf H_U^{(k)}=h_{\rm w}\mathsf D_k^T\mathsf D_k\) and
\(\det\mathsf H_U^{(0)}=1\).  This common value is precisely the part of
the \(k\)th delay that cannot be reproduced by the preceding directions.
The chronological order is used only to resolve the information contributed
by successive delays; simultaneously permuting the delays and their weights
leaves the GM cost unchanged.
\end{proof}

\begin{theorem}[Single-qubit GM bound]
Consider a single-qubit reservoir driven by a pure input qubit.  Assume that
the operating-point fixed state is full rank and that the three-delay QFIM is
positive definite.  For a fixed joint unitary \(U\), let \(A=A(U)\) be the
linear Bloch part of the reference channel and set
\(\kappa_U=\sqrt{|\det A|}\).  For every positive diagonal weight
\(\Omega=\operatorname{diag}(\omega_1,\omega_2,\omega_3)\),
\begin{equation}
    C_{\rm GM}(\Omega,\mathsf H_U)
    \ge f_\Omega(\kappa_U),
    \qquad
    f_\Omega(\kappa)
    =
    \frac{
        \left(
            \sqrt{\omega_1}
            +\sqrt{\omega_2}
            +\sqrt{\omega_3}/\kappa
        \right)^2
    }{
        \pi^2(1-\kappa^2)
    }.
    \label{eq:S_single_GM_bound}
\end{equation}
The Fisher cycle in Sec.~\ref{sec:S_explicit_single_cycle} attains the bound
with \(\kappa_U=r\).
\end{theorem}

\begin{proof}
Positive definiteness implies \(0<\kappa<1\); otherwise the rank-two spectrum
or the writing budget makes the three-delay QFIM singular.  We suppress the
subscript on \(\kappa=\kappa_U\) below.  Because the reservoir is a qubit, the
GM dimension factor is one.  With
\(\|\mathsf K\|_\ast
\equiv\operatorname{Tr}\sqrt{\mathsf K^\dagger\mathsf K}\), one has
\begin{equation}
    \begin{aligned}
    \sqrt{C_{\rm GM}(\Omega,\mathsf H_U)}
    &=\left\|\mathsf C^{-T}\Omega^{1/2}\right\|_\ast\\
    &\ge
    \left|\operatorname{Tr}
    (\mathsf C^{-T}\Omega^{1/2})\right|
    =\sum_{k=1}^3\frac{\sqrt{\omega_k}}{\mathsf C_{kk}}.
    \end{aligned}
    \label{eq:S_permuted_GM_bound}
\end{equation}
The norm identity follows from \(\mathsf H_U=\mathsf C^T\mathsf C\).  We then
use the elementary trace bound
\begin{equation}
    \|\mathsf B\|_\ast
    =\max_{\mathsf U\ {\rm unitary}}
    \left|\Tr(\mathsf U^\dagger\mathsf B)\right|
    \ge |\Tr\mathsf B|,
    \label{eq:S_nuclear_trace_bound}
\end{equation}
with \(\mathsf B=\mathsf C^{-T}\Omega^{1/2}\).  The final equality uses the
positive diagonal of this lower-triangular matrix.

\textbf{The three lemmas reduce the full channel optimization to a scalar
tradeoff controlled by \(\kappa\): fresh writing, retention of the first
delayed tangent, and retention of the two-dimensional area spanned by the
first two delayed tangents cannot be optimized independently.}  They give the
single chain
\begin{equation}
    \begin{aligned}
    \sqrt{C_{\rm GM}}
    &\ge
    \frac{
    \sqrt{\omega_1}+\sqrt{\omega_2}/\nu_1
    +\sqrt{\omega_3}/\kappa
    }{\sqrt{h_{\rm w}}}\\
    &\ge
    \frac{
    \sqrt{\omega_1}+\sqrt{\omega_2}
    +\sqrt{\omega_3}/\kappa
    }{\pi\sqrt{1-\kappa^2}}.
    \end{aligned}
    \label{eq:S_fixed_channel_GM_bound}
\end{equation}
Here the second line uses \(\nu_1\le1\) and
Eq.~\eqref{eq:S_single_write_bound}.  Squaring proves
Eq.~\eqref{eq:S_single_GM_bound}.

For the proposed cycle, \(\widehat A_{\rm cyc}=A_{\rm cyc}\),
\(A_{\rm cyc}^3=r^2\id_3\), and
\(\operatorname{sv}(A_{\rm cyc})=(1,r,r)\).  Its normalized delayed
directions are \((\bm e_X,\bm e_Y,r\bm e_Z)\), and
\(h_{\rm w}=\pi^2(1-r^2)\).  Hence
\begin{equation}
    \mathsf H_{\rm cyc}^{(3)}(r)
    =
    \pi^2(1-r^2)\operatorname{diag}(1,1,r^2),
    \qquad
    C_{\rm GM}\!\left(\Omega,\mathsf H_{\rm cyc}^{(3)}(r)\right)
    =
    f_\Omega(r).
    \label{eq:S_single_Hcyc_opt}
\end{equation}
Thus every inequality is saturated when \(r=\kappa\).
\end{proof}

Because the cycle realizes every \(0<\kappa<1\), the global optimum is
\begin{equation}
    \inf_U C_{\rm GM}(\Omega,\mathsf H_U)
    =
    \min_{0<r<1}f_\Omega(r).
    \label{eq:S_single_global_optimum}
\end{equation}
The main-text endpoint losses, for which one or more weights vanish, follow
continuously from this positive-weight result.

For uniform weighting,
\begin{equation}
    f_{\rm unif}(r)
    =
    \frac{(2+r^{-1})^2}{3\pi^2(1-r^2)},
    \qquad
    2r_{\rm opt}^3+2r_{\rm opt}^2-1=0,
    \qquad
    r_{\rm opt}\simeq0.56520,
    \qquad
    C_{\rm GM}^{\rm unif}\simeq0.70508.
\end{equation}

\subsubsection{Equality conditions for the optimal cycle}

The theorem above shows that the Fisher cycle attains the global bound.  We
now ask whether a channel with a qualitatively different tangent dynamics can
attain the same value.  For strictly positive weights, equality in the final
bound requires every inequality used in its derivation to be saturated.
Following these equality conditions shows that orthogonal cyclic routing is
not merely sufficient, but necessary at the optimum.

\textbf{First, the writing and retention bounds must both be saturated.}
Equality throughout Eq.~\eqref{eq:S_fixed_channel_GM_bound} requires
\begin{equation}
    \nu_1=1,
    \qquad
    h_{\rm w}=\pi^2(1-\kappa^2),
    \qquad
    \mathsf C_{22}=\mathsf C_{11},
    \qquad
    \mathsf C_{33}=\kappa\mathsf C_{11}.
    \label{eq:S_single_equality_conditions}
\end{equation}
Thus the channel uses the full fresh-writing budget, retains one delayed
direction without loss, and retains the third independent direction with the
maximal factor \(\kappa\).

\textbf{Second, the three delays cannot overlap in Fisher space.}
The nuclear-norm inequality in Eq.~\eqref{eq:S_permuted_GM_bound} must also
be saturated.  Let
\(\mathsf B=\mathsf C^{-T}\Omega^{1/2}\).  Equality in
\(\|\mathsf B\|_\ast\ge\operatorname{Tr}\mathsf B\) requires the real matrix
\(\mathsf B\) to be symmetric and positive semidefinite.  On the other hand,
\(\mathsf B\) is lower triangular with positive diagonal.  It must therefore
be diagonal.  Since every \(\omega_k>0\), this also makes \(\mathsf C\)
diagonal.  By the chronological QR construction, its off-diagonal entries
are precisely the projections of later delays onto earlier Fisher directions.
Their disappearance means that the three delayed directions are mutually
Fisher orthogonal.

Normalize these directions as
\[
    \bm e_1=\bm d_1,
    \qquad
    \bm e_2=\bm d_2,
    \qquad
    \bm e_3=\bm d_3/\kappa.
\]
Equations~\eqref{eq:S_causal_QR} and
\eqref{eq:S_single_equality_conditions} make
\((\bm e_1,\bm e_2,\bm e_3)\) an orthonormal Fisher basis and give
\[
    \widehat A_\ast\bm e_1=\bm e_2,
    \qquad
    \widehat A_\ast\bm e_2=\kappa\bm e_3.
\]

\textbf{Third, the rank-two singular spectrum closes the route into a
three-cycle.}  The condition \(\nu_1=1\) gives
\(\operatorname{sv}(\widehat A_\ast)=(1,\kappa,\kappa)\).  Because
\(\widehat A_\ast\bm e_1=\bm e_2\) and both vectors are normalized,
\(\bm e_1\) and \(\bm e_2\) are the corresponding right and left unit
singular directions.  The restriction of \(\widehat A_\ast\) to the
orthogonal complement of \(\bm e_1\) is therefore \(\kappa\) times an
isometry.  Since it already maps
\(\bm e_2\) to \(\kappa\bm e_3\), the image of the remaining basis vector
must be orthogonal to both \(\bm e_2\) and \(\bm e_3\).  Hence
\(\widehat A_\ast\bm e_3=\pm\kappa\bm e_1\), and
\begin{equation}
    [\widehat A_\ast]_{\{\bm e_j\}}
    =
    \begin{pmatrix}
        0&0&\pm\kappa\\
        1&0&0\\
        0&\kappa&0
    \end{pmatrix},
    \qquad
    \widehat A_\ast^3=\pm\kappa^2\id_3.
    \label{eq:S_single_saturating_cycle}
\end{equation}
Finally, \(\nu_1=1\) makes the canonical translation vanish through
Eq.~\eqref{eq:S_rank_two_parameters}; every saturating reference channel is
therefore unital.  \textbf{Every saturating channel consequently realizes
the same damped three-cycle in the Fisher-normalized tangent space, up to
rotations of the Fisher coordinate system and sign conventions.}  This fixes
the optimal induced reservoir dynamics; it does not require a unique
microscopic joint-unitary realization.

\subsubsection{Comparison with amplitude-damping channels}

This comparison concerns only the most favorable contraction spectrum that
amplitude damping could supply after an idealized cyclic routing of its
principal directions; it is not a competing full-rank stationary QRC model.
An amplitude-damping channel with determinant parameter \(\kappa\) has
\(\operatorname{sv}(\widehat A_{\rm AD})
=(\sqrt{\kappa},\sqrt{\kappa},\kappa)\).  After cyclically routing its three
principal contraction directions, the most favorable three-delay QFIM and
its GM bound are
\begin{equation}
    \begin{aligned}
    C_{\rm GM}(\Omega,\mathsf H_{\rm AD})
    &\ge
    \frac{
        \left(
            \sqrt{\omega_1}
            +\sqrt{\omega_2/\kappa}
            +\sqrt{\omega_3}/\kappa
        \right)^2
    }{
        \pi^2(1-\kappa^2)
    },\\
    \mathsf H_{\rm AD}^{(3)}(\kappa)
    &=
    \pi^2(1-\kappa^2)
    \operatorname{diag}(1,\kappa,\kappa^2).
    \end{aligned}
    \label{eq:S_single_H_AD}
\end{equation}
At equal determinant parameter \(\kappa=r\), amplitude damping therefore
pays an additional factor \(\kappa\) in the second-delay QFI.  For uniform
weighting,
\begin{equation}
    C_{\rm AD}^{\rm unif}(\kappa)
    =
    \frac{
        (1+\kappa^{-1/2}+\kappa^{-1})^2
    }{
        3\pi^2(1-\kappa^2)
    },
    \qquad
    \kappa_{\rm AD,opt}\simeq0.60668,
    \qquad
    C_{\rm AD}^{\rm unif}\simeq0.82636.
\end{equation}
This spectral benchmark is about \(17.2\%\) above the Fisher-cycle optimum.
Moreover, bare amplitude damping has a pure fixed point, so its actual
three-delay stationary QFIM is rank deficient.  The endpoints \(\kappa=0\) and
\(\kappa=1\) are likewise degenerate and do not affect the interior optimum.

\section{Clifford Pauli Routes for Fisher-Orthogonal Memory}
\label{sec:S_many_body}

This section develops the many-body reservoir and establishes the properties
needed for temporal processing.  We first derive how a fresh input tangent is
written, damped, and transported through mutually orthogonal Pauli directions.
We then address three questions that are not evident from the circuit alone:
how long the Pauli route remains distinct, whether its binary dynamics is
irreducible or decomposes into invariant components, and which delayed
products survive at second order.

Clifford conjugation turns these dynamical questions into the iteration of a
binary symplectic matrix.  \textbf{The polynomial analysis below has two
specific purposes.}  It certifies the closure and exact period of the Pauli
route, and it determines whether the route can be confined to a proper Pauli
subspace.  The first property fixes the designed memory window; the second
determines which argument can be used in the later analysis of the echo-state
property.

\subsection{Pauli-route dynamics of the QRC channel}

We begin with the driven channel because it identifies precisely what the
later Clifford construction must provide.  For an \(N\)-qubit reservoir of
dimension \(d=2^N\), the Pauli string \(Q\) specifies the freshly written
direction, \(G\) determines which stored directions are damped, and the
reservoir Clifford \(V\) advances the recurrent route.  Let \(Q\) and \(G\)
satisfy \(Q^2=G^2=\id\) and \(\{Q,G\}=0\).  Define
\begin{equation}
    W_r(Q,G)
    =
    \begin{pmatrix}
        a_r\id_{\rm R} & b_r Q\\
        b_r G & -a_rGQ
    \end{pmatrix}_{\rm in},
    \qquad
    a_r=\sqrt{\frac{1+r}{2}},\quad
    b_r=\sqrt{\frac{1-r}{2}}.
    \label{eq:S_many_W}
\end{equation}
The complete driven step is
\begin{equation}
    U_{r,V,Q,G}
    =
    (\id_{\rm in}\otimes V)W_r(Q,G)
    (L_{\rm in}\otimes\id_{\rm R}),
    \qquad
    L_{\rm in}=e^{+\ii\pi Y/4}.
    \label{eq:S_many_U}
\end{equation}
We denote the corresponding reduced reservoir channel by
\(\Phi_{s,r}^{V,Q,G}\).
The Pauli identities \(Q^2=G^2=\id_{\rm R}\) make \(W_r(Q,G)\) a unitary block.
The additional anticommutation condition \(GQG=-Q\) realizes the desired
write--store structure: the fresh input derivative is written into the
\(Q\) direction with the same full-strength coefficient as in the
single-qubit cycle, while \(Q\) lies in a sector damped by \(G\).  Using the
rotated input amplitudes \(\ell_s,m_s\) from Eq.~\eqref{eq:S_rotated_input},
the write--store channel before the final routing has Kraus operators
\begin{equation}
    B_0(s)=\ell_s a_r\id_{\rm R}+m_s b_rQ,
    \qquad
    B_1(s)=\ell_s b_rG-m_s a_rGQ,
    \label{eq:S_many_kraus}
\end{equation}
and
\(\Psi_{s,r}^{Q,G}(\rho)
=\sum_{\nu=0,1}B_\nu(s)\rho B_\nu^\dagger(s)\).  At
\(s_0=1/2\), these reduce to \(B_0=a_r\id_{\rm R}\), \(B_1=b_rG\), with
\(\dot B_0=(\pi/2)b_rQ\) and \(\dot B_1=-(\pi/2)a_rGQ\).  Therefore the
operating-point channel is
\begin{equation}
    \Psi_{s_0,r}^{Q,G}(\rho)
    =
    \frac{1+r}{2}\rho+\frac{1-r}{2}G\rho G.
    \label{eq:S_many_Psi}
\end{equation}
Consequently,
\begin{equation}
    \Psi_{s_0,r}^{Q,G}(P)
    =
    \begin{cases}
        P,& [G,P]=0,\\
        rP,& \{G,P\}=0.
    \end{cases}
    \label{eq:S_many_damping}
\end{equation}

Differentiating with respect to \(s\) at \(s_0\) gives the linear-response
superoperator.  Keeping the adjoint of \(GQ\) explicit,
\(\dot B_1^\dagger=-(\pi/2)a_rQG\), one obtains
\begin{equation}
    \begin{aligned}
    \dot\Psi_{s_0,r}^{Q,G}(\rho)
    =
    &\sum_{\nu=0,1}
    \left[
        \dot B_\nu(s_0)\rho B_\nu^\dagger(s_0)
        +
        B_\nu(s_0)\rho\dot B_\nu^\dagger(s_0)
    \right] \\
    =
    &\frac{\pi\sqrt{1-r^2}}{4}
    \left(
        Q\rho+\rho Q-GQ\rho G-G\rho QG
    \right).
    \end{aligned}
    \label{eq:S_many_dotPsi_general}
\end{equation}
For \(\rho_\ast=\id_{\rm R}/d\), this gives the information-writing derivative
\begin{equation}
    \dot\Psi_{s_0,r}^{Q,G}(\rho_\ast)
    =
    \frac{\pi\sqrt{1-r^2}}{4d}
    \left(2Q-2GQG\right)
    =
    \frac{\pi\sqrt{1-r^2}}{d}Q,
    \label{eq:S_many_write_derivative}
\end{equation}
where the last equality uses \(\{Q,G\}=0\).  This is the many-body analogue of
Eq.~\eqref{eq:S_single_injection}.  If \(Q\) commuted with \(G\), the fresh
write derivative would cancel exactly because \(GQG=Q\).  Thus the
anticommutation condition is necessary for information writing in this block,
not just for the later damping pattern.  After the final routing,
\begin{equation}
    \dot\Phi_{s_0,r}^{V,Q,G}(\rho_\ast)
    =
    \frac{\pi\sqrt{1-r^2}}{d}\,VQV^\dagger.
    \label{eq:S_many_fresh_tangent}
\end{equation}
Let the target route be
\begin{equation}
    P_{j+1}=VP_jV^\dagger,\qquad P_{K+1}=P_1,
    \qquad P_1=VQV^\dagger.
    \label{eq:S_many_route}
\end{equation}
Define the hit word
\begin{equation}
    \eta_j=
    \begin{cases}
        0,& [G,P_j]=0,\\
        1,& \{G,P_j\}=0,
    \end{cases}
    \qquad
    n_m(\eta)=\sum_{j=1}^m\eta_j,\quad n_0=0.
    \label{eq:S_eta_word}
\end{equation}
For a \(K\)-periodic route we write
\(\eta_K(V,Q,G)=(\eta_1,\ldots,\eta_K)\), with the same accumulated count
\(n_m(\eta_K)\).
A perturbation injected \(k\) steps before readout has tangent
\begin{equation}
    \Delta_k
    =
    \frac{\pi\sqrt{1-r^2}}{d}\,
    r^{n_{k-1}(\eta_K)}P_k,
    \label{eq:S_many_Delta_k}
\end{equation}
up to an irrelevant sign.  Using Eq.~\eqref{eq:S_maxmixed_metric}, the
designed \(K\)-delay QFIM is
\begin{equation}
    \left(\mathsf H^{(K)}_{\eta_K}(r)\right)_{kl}
    =
    \pi^2(1-r^2)r^{2n_{k-1}(\eta_K)}\delta_{kl}.
    \label{eq:S_many_H_eta}
\end{equation}
Thus \(V\) determines the operator direction assigned to each delay, whereas
the binary word \(\eta_K\) records only the damping accumulated before that
direction reaches the readout.  Orthogonality of distinct Pauli strings makes
the designed QFIM diagonal.  \textbf{Routing fixes where each delay is stored;
the hit word fixes only how strongly that direction survives.}  The remaining
task is therefore to construct a
long route of distinct Pauli strings and characterize its invariant-subspace
structure, which will enter the echo-state analysis below.

\subsection{Binary symplectic representation of Clifford routes}

For an \(N\)-qubit reservoir, we ignore Pauli phases and represent each Pauli
string by
\[
    \bm v=(x_1,\ldots,x_N,z_1,\ldots,z_N)^T\in\GF^{2N}.
\]
We write \(P(\bm v)\) for the corresponding phase-free Pauli string.
Here \(\GF=\{0,1\}\) is the binary field.  All additions and multiplications
are evaluated modulo \(2\), so \(1+1=0\) and subtraction is the same as
addition.  The term ``field'' only means that the usual arithmetic operations,
including division by a nonzero element, remain within this set.  On each
qubit,
\[
    I\leftrightarrow(0,0),\qquad
    X\leftrightarrow(1,0),\qquad
    Z\leftrightarrow(0,1),\qquad
    Y\leftrightarrow(1,1).
\]
The binary symplectic form is
\begin{equation}
    \symp{\bm u}{\bm v}=\bm u^T\mathsf J\bm v,
    \qquad
    \mathsf J=
    \begin{pmatrix}
        0&\id_N\\
        \id_N&0
    \end{pmatrix},
    \label{eq:S_symplectic_form}
\end{equation}
and two Pauli strings anticommute iff
\(\symp{\bm u}{\bm v}=1\).  A Clifford unitary
\(V\) acts on these binary labels through a \(2N\times2N\) symplectic matrix
\(M\) \cite{SMHostens2005Clifford}:
\begin{equation}
    V P(\bm v)V^\dagger=\pm P(M\bm v),
    \qquad
    M^T\mathsf J M=\mathsf J.
    \label{eq:S_Clifford_symplectic}
\end{equation}
Thus \(V\) is the physical reservoir unitary, whereas \(M\) is its binary
representation on Pauli strings.  The signs do not affect the Fisher matrix
or the damping words.  Our row and column convention follows directly from
the ordering of \(\bm v\).  Let \(\bm e_j\) denote the \(j\)th binary unit
vector.
Then
\[
    P(\bm e_j)=X_j,\qquad
    P(\bm e_{N+j})=Z_j,
    \qquad 1\le j\le N.
\]
The \(j\)th column of \(M\) is therefore the binary label of
\(VX_jV^\dagger\), and the \((N+j)\)th column is the label of
\(VZ_jV^\dagger\).  Within each column, the first \(N\) rows give the output
\(X\) components and the last \(N\) rows give the output \(Z\) components.
This convention makes the complete high-dimensional matrix recoverable
from the Clifford action on the \(2N\) Pauli generators.

An important advantage of this representation is that several global
properties of the Clifford routing are encoded by ordinary binary
polynomials, most notably the characteristic polynomial of \(M\).  We first
list the concrete matrices used in this work and then use their
characteristic polynomials to determine the corresponding Pauli-cycle
periods.

Before discussing the general period criterion, we give the concrete routing
generators used throughout this work.  They will also serve as examples for
the abstract properties below.  Pauli strings are ordered as
\(\sigma^{(1)}\cdots\sigma^{(N)}
=\sigma^{(1)}\otimes\cdots\otimes\sigma^{(N)}\), binary labels
use \(\bm v=(x_1,\ldots,x_N,z_1,\ldots,z_N)^T\), and products of gates are
applied from right to left.  Let
\[
    H_{\rm Had}
    =
    \frac{1}{\sqrt2}
    \begin{pmatrix}
        1&1\\
        1&-1
    \end{pmatrix},
    \qquad
    S=
    \begin{pmatrix}
        1&0\\
        0&\ii
    \end{pmatrix},
\]
and let \(C_{a\to b}\) denote a controlled-NOT (CNOT) gate with control \(a\)
and target \(b\).
Global phases are immaterial.

The one-qubit routing generator, denoted by \(V_{\rm cyc}\) in
Sec.~S3 and by \(V_3\) when its period is emphasized, is, up to a global
phase,
\begin{equation}
    V_3
    =
    \frac{1}{2}\left(I-\ii X-\ii Y-\ii Z\right),
    \qquad
    M_3=
    \begin{pmatrix}
        1&1\\
        1&0
    \end{pmatrix}.
    \label{eq:S_explicit_V3}
\end{equation}
It realizes \(X\to Y\to Z\to X\).

The two-qubit routing generator used as the running example below is
\begin{equation}
    V_5
    =
    C_{1\to2}\,
    S_1\,
    (H_{\rm Had})_2\,
    (H_{\rm Had})_1,
    \label{eq:S_explicit_V5}
\end{equation}
with binary symplectic matrix
\begin{equation}
    M_5=
    \begin{pmatrix}
        0&0&1&0\\
        0&0&1&1\\
        1&1&1&0\\
        0&1&0&0
    \end{pmatrix}.
    \label{eq:S_explicit_M5}
\end{equation}
Ignoring Pauli signs, its action on the four single-site generators is
\begin{equation}
\begin{aligned}
    X_1&\longmapsto Z_1,
    &
    Z_1&\longmapsto Y_1X_2,
    \\
    X_2&\longmapsto Z_1Z_2,
    &
    Z_2&\longmapsto X_2 .
\end{aligned}
    \label{eq:S_M5_generator_action}
\end{equation}
Since conjugation preserves products, Eq.~\eqref{eq:S_M5_generator_action}
determines the image of every two-qubit Pauli string and allows the
period-five routes in Fig.~2 of the main text to be checked directly.

For the three-qubit generator that produces the period-nine Singer route, we
use
\begin{equation}
\begin{aligned}
    V_9
    &=
    C_{2\to1}\,
    C_{3\to1}\,
    C_{3\to2}\,
    (H_{\rm Had})_2\,
    C_{2\to3}\,
    S_1\,
    S_3\,
    (H_{\rm Had})_3\,
    S_3\,
    (H_{\rm Had})_2 ,
\end{aligned}
    \label{eq:S_explicit_V9}
\end{equation}
whose symplectic matrix is
\begin{equation}
    M_9=
    \begin{pmatrix}
        1&1&0&0&0&1\\
        0&1&1&0&1&0\\
        0&0&1&0&1&1\\
        1&0&0&1&0&0\\
        1&0&0&1&1&0\\
        1&0&0&1&1&1
    \end{pmatrix}.
    \label{eq:S_explicit_M9}
\end{equation}

\subsection{Singer-route period and complete Pauli transport}
\label{sec:S_singer_roots}

The circuit representation specifies a Clifford unitary but does not by
itself determine the length of its Pauli route or whether the associated
binary dynamics is irreducible.  The first property fixes the number of
distinct Fisher-orthogonal memory directions.  The second determines whether
the Singer spanning argument can be used in the later echo-state analysis.
Both can be read compactly from the characteristic polynomial, without
enumerating all \(4^N-1\) nonidentity Pauli strings.

\textbf{The polynomial plays two different roles.}  Polynomial divisibility
gives a candidate closing period, and an additional order test determines
whether that period is minimal.  Irreducibility instead excludes a proper
invariant Pauli subspace.  The resulting spanning property is not an
independent performance objective; it is the algebraic input required by the
Singer-route echo-state proof in Sec.~\ref{sec:S_esp}.

For a binary symplectic matrix \(M\), define
\begin{equation}
    \chi_M(\lambda)
    =
    \det(\lambda\id_{2N}-M)
    =
    \lambda^{2N}+c_{2N-1}^{(M)}\lambda^{2N-1}
    +\cdots+c_1^{(M)}\lambda+1 .
    \label{eq:S_characteristic_polynomial}
\end{equation}
Here \(\lambda\) is only a formal polynomial variable, not the eigenvalue
notation \(\lambda_a\), \(\lambda_{\min}\), or \(\lambda_{\max}\) used
elsewhere, and all coefficients belong to \(\GF=\{0,1\}\).
Throughout this subsection, additions between polynomial
coefficients and between binary vectors or matrices are understood modulo
two; exponents, dimensions, and orbit lengths are ordinary integers.  We
retain the conventional \(+\) notation inside polynomials and write
\(\lambda^K-1\) for the periodicity condition, noting that \(1=-1\) in \(\GF\).
The constant term is one because \(M\) is invertible.  Two root symmetries
hold before any irreducibility assumption is made.
First, \(M^{-1}=\mathsf J^{-1}M^T\mathsf J\) implies
\begin{equation}
    \chi_M(\lambda)=\lambda^{2N}\chi_M(\lambda^{-1}),
    \label{eq:S_reciprocal_characteristic}
\end{equation}
so \(\chi_M\) is reciprocal: if \(\zeta\) is a root in a field where the
polynomial factorizes, then \(\zeta^{-1}\) is also a root.  Second,
arithmetic modulo two gives
\begin{equation}
    \chi_M(\zeta^2)=\bigl[\chi_M(\zeta)\bigr]^2=0,
    \label{eq:S_root_squaring}
\end{equation}
so the square of every root is again a root.  The scalar \(\zeta\) is a root
value, not the matrix \(M\).

Now suppose that \(\chi_M\) is irreducible of degree \(2N\).  Over
\(\GF\), its roots are the \(2N\) distinct conjugates obtained by repeated
squaring,
\(\zeta,\zeta^2,\ldots,\zeta^{2^{2N-1}}\).  Reciprocity requires
\(\zeta^{-1}\) to be one of these same roots.  The standard theorem for
irreducible self-reciprocal polynomials places this inverse after \(N\)
squarings \cite{SMMeynGotz1989SelfReciprocal}:
\[
    \zeta^{-1}=\zeta^{2^N},
    \qquad
    \zeta^{2^N+1}=1.
\]
Hence \(\zeta\) is a root of \(\lambda^{2^N+1}-1\).  Since
\(\chi_M(\lambda)\) is the
minimal polynomial of \(\zeta\) over \(\GF\), it follows that
\begin{equation}
    \chi_M(\lambda)\mid \lambda^K-1,
    \qquad
    K=2^N+1.
    \label{eq:S_char_divisibility}
\end{equation}
Cayley--Hamilton then yields
\begin{equation}
    M^K=\id_{2N}.
    \label{eq:S_reciprocal_period_bound}
\end{equation}
Thus irreducibility and symplectic reciprocity imply only that the matrix
order divides \(K\).  We say that the corresponding Clifford generator
produces a Singer cycle when \(M\) has order exactly \(K\).
In what follows, ``Singer cycle'' refers to this algebraic action of \(M\),
whereas ``Singer route'' denotes a particular Pauli-label orbit generated by
that action.  Once
Eq.~\eqref{eq:S_reciprocal_period_bound} is known, exactness is checked by
\begin{equation}
    M^{K/p}\ne\id_{2N}
    \qquad
    \text{for every prime divisor }p\text{ of }K .
    \label{eq:S_singer_test}
\end{equation}

The routing matrices used in this work make the distinction explicit.  Their
characteristic polynomials satisfy
\begin{equation}
\begin{aligned}
    \chi_{M_3}(\lambda)&=\lambda^2+\lambda+1,
    &\lambda^3-1&=(\lambda-1)\chi_{M_3}(\lambda),\\
    \chi_{M_5}(\lambda)&=\lambda^4+\lambda^3+\lambda^2+\lambda+1,
    &\lambda^5-1&=(\lambda-1)\chi_{M_5}(\lambda),\\
    \chi_{M_9}(\lambda)&=\lambda^6+\lambda^3+1,
    &\lambda^9-1&=(\lambda^3-1)\chi_{M_9}(\lambda).
\end{aligned}
    \label{eq:S_explicit_characteristic_polynomials}
\end{equation}
All three polynomials are irreducible and reciprocal.  Since three and five
are prime, \(M_3\ne\id_2\) and \(M_5\ne\id_4\) fix their exact orders.  For
\(M_9\), the only nontrivial candidate is three, and direct evaluation gives
\(M_9^3\ne\id_6\).  Hence \(M_3,M_5\), and \(M_9\) generate Singer cycles of
periods three, five, and nine.  The \(M_5\) example also makes the two root
symmetries above concrete.  Its period gives
\begin{equation}
    M_5^{-1}=M_5^4,
    \label{eq:S_M5_inverse_period}
\end{equation}
which is the matrix counterpart of
\(\zeta^{-1}=\zeta^{2^N}=\zeta^4\) at \(N=2\).  Since both \(2\) and \(4\)
are coprime to \(5\), \(M_5^2\) and \(M_5^4\) visit the same five Pauli
labels as \(M_5\), but in a different temporal order.

Irreducibility has a direct consequence for Pauli routing.  For any nonzero
label \(\bm v\), define
\begin{equation}
    \mathcal W_{\bm v}
    =
    \operatorname{span}_{\GF}
    \{\bm v,M\bm v,M^2\bm v,\ldots\}.
    \label{eq:S_orbit_span}
\end{equation}
This is a nonzero \(M\)-invariant subspace.  Suppose, to the contrary, that it
does not span the full label space, and write
\(d_{\mathcal W}=\dim\mathcal W_{\bm v}\), with
\(0<d_{\mathcal W}<2N\).  Extending a basis of
\(\mathcal W_{\bm v}\) to the full label space puts \(M\) in
block-upper-triangular form because
\(M\mathcal W_{\bm v}\subseteq\mathcal W_{\bm v}\):
\[
    M=
    \begin{pmatrix}
        M_{11}&M_{12}\\
        0&M_{22}
    \end{pmatrix},
    \qquad
    M_{11}\in\GF^{d_{\mathcal W}\times d_{\mathcal W}},
    \quad
    M_{22}\in
    \GF^{(2N-d_{\mathcal W})\times(2N-d_{\mathcal W})}.
\]
The determinant of a block-triangular matrix is the product of the
determinants of its diagonal blocks.  Therefore
\[
    \chi_M(\lambda)
    =
    \det(\lambda\id_{d_{\mathcal W}}-M_{11})\,
    \det(\lambda\id_{2N-d_{\mathcal W}}-M_{22}),
\]
which is a nontrivial factorization because both factors have positive
degree.  This contradicts the irreducibility of \(\chi_M\).  Hence the full
orbit spans \(\mathcal W_{\bm v}=\GF^{2N}\).  Cayley--Hamilton expresses every
higher power through the first \(2N\), so these first \(2N\) labels already
span the \(2N\)-dimensional space.  Because there are exactly \(2N\) of them,
\begin{equation}
    \bm v,M\bm v,\ldots,M^{2N-1}\bm v
    \quad\text{form a basis of }\GF^{2N}.
    \label{eq:S_singer_orbit_basis}
\end{equation}
Since \(M\) is invertible, any \(2N\) consecutive labels on the same orbit
also form a basis.

We can now determine the length of every nonzero Pauli orbit.  Since
\(M^K=\id_{2N}\), every label returns after \(K\) steps.  Suppose that a
nonzero label \(\bm v\) returned earlier, so that
\(M^m\bm v=\bm v\) for some \(0<m<K\).  The fixed subspace
\[
    \mathcal F_m=\ker(M^m-\id_{2N})
\]
would then be nonzero.  It is also invariant under \(M\): if
\(\bm w\in\mathcal F_m\), then
\(M^m(M\bm w)=M(M^m\bm w)=M\bm w\), and hence
\(M\bm w\in\mathcal F_m\).  Because \(\chi_M\) is irreducible, \(M\) has no
nonzero proper invariant subspace.  Therefore
\(\mathcal F_m=\GF^{2N}\), which would mean that
\(M^m=\id_{2N}\) on the full label space.  This contradicts the assumption
that \(M\) has exact order \(K\).  Thus every nonzero Pauli label has orbit
length exactly \(K=2^N+1\).

There are \(2^{2N}-1=4^N-1\) nonzero binary labels, one for each
nonidentity Pauli string.  Since distinct orbits do not overlap and each
contains \(K\) labels, they split into
\begin{equation}
    \frac{4^N-1}{2^N+1}=2^N-1
    \label{eq:S_number_singer_orbits}
\end{equation}
disjoint period-\(K\) cycles \cite{SMHestenes1970SingerGroups}.  Equivalently,
the nonidentity Pauli strings form \(2^N-1\) separate routes, each visiting
\(2^N+1\) distinct operator directions before returning to its starting
point.

The same Pauli strings also admit a complementary partition into commuting
classes.  For the Singer generators used here, \(K=2^N+1\) maximal commuting
Pauli classes are permuted cyclically.  Each class contains \(2^N-1\)
nonidentity strings and, together with the identity, forms a maximal Abelian
Pauli subgroup of size \(2^N\).  Distinct classes share no nonidentity
element.  Since \((2^N+1)(2^N-1)=4^N-1\), they exhaust all nonidentity Pauli
strings.  Each of the \(2^N-1\) period-\(K\) orbits above therefore intersects
every commuting class exactly once.  Equivalently, the Pauli strings form a
\((2^N-1)\times(2^N+1)\) array whose rows are Singer orbits and whose columns
are maximal commuting classes.  Every column can therefore be accessed by one
joint Pauli measurement.  The same partition is known as the cyclic
mutually-unbiased-basis construction \cite{SMKern2010CyclicMUB} and gives the
row--column structure illustrated in Fig.~2 of the main text.

\subsection{Composite Clifford routes}

Suppose the reservoir is partitioned into blocks, with a local Clifford
routing generator \(V_\mu\) producing a period-\(K_\mu\) route on block
\(\mu\).  The product generator and the resulting global period obey
\begin{equation}
    V_{\rm comp}=\bigotimes_\mu V_\mu,
    \qquad
    K_{\rm comp}=\operatorname{lcm}(K_1,K_2,\ldots).
    \label{eq:S_product_clock}
\end{equation}
For each block, let \(\eta^{(\mu)}\in\{0,1\}^{K_{\rm comp}}\) be its local
hit word, periodically extended to the common period.  For factorized Pauli
strings, the global word is the elementwise parity
\begin{equation}
    \eta_{\rm comp}(Q_{\rm comp},G_{\rm comp})
    =
    \bigoplus_\mu\eta^{(\mu)}(Q_\mu,G_\mu).
    \label{eq:S_product_eta}
\end{equation}
Here \(\bigoplus\) denotes elementwise addition modulo two, not the number
of damping events accumulated along the route.
For
\(Q_{\rm comp}=\bigotimes_\mu Q_\mu\) and
\(G_{\rm comp}=\bigotimes_\mu G_\mu\), let
\(\bm v_{Q,\mu}\) and \(\bm v_{G,\mu}\) denote the binary labels of the
component strings \(Q_\mu\) and \(G_\mu\).  The same parity rule gives
\begin{equation}
    \{Q_{\rm comp},G_{\rm comp}\}=0
    \quad\Longleftrightarrow\quad
    \bigoplus_\mu
    \symp{\bm v_{Q,\mu}}{\bm v_{G,\mu}}
    =1.
    \label{eq:S_product_write_condition}
\end{equation}
Equation~\eqref{eq:S_product_write_condition} guarantees
\(\{Q_{\rm comp},G_{\rm comp}\}=0\) and hence a nonzero fresh write-in branch
at the operating point.  It does not ensure that every previously stored mode
is damped; that separate requirement is analyzed in
Sec.~\ref{sec:S_esp}.  By contrast, the periods \(K_\mu\), their least common
multiple, and the products below are ordinary integers.

The period-fifteen construction used in the main text is the simplest
example.  Its recurrent unitary and binary representation are
\begin{equation}
    V_{15}=V_5^{(1,2)}\otimes V_3^{(3)},
    \qquad
    M_{15}\simeq M_5\oplus M_3,
    \label{eq:S_explicit_V15}
\end{equation}
where \(\simeq\) denotes the coordinate permutation from
\((x_1,x_2,x_3,z_1,z_2,z_3)\) to
\((x_1,x_2,z_1,z_2,x_3,z_3)\).  Hence
\begin{equation}
    \operatorname{ord}(M_{15})=\operatorname{lcm}(5,3)=15,
    \qquad
    \chi_{M_{15}}(\lambda)
    =\chi_{M_5}(\lambda)\chi_{M_3}(\lambda).
    \label{eq:S_explicit_M15}
\end{equation}
Explicitly,
\begin{equation}
    \chi_{M_{15}}(\lambda)
    =
    (\lambda^4+\lambda^3+\lambda^2+\lambda+1)
    (\lambda^2+\lambda+1).
    \label{eq:S_M15_characteristic}
\end{equation}
Thus \(M_{15}\) generates a period-\(15\) composite route rather than a
six-dimensional Singer cycle.  Although its period follows directly from the
component periods, the global matrix is reducible, so the spanning argument
used for a single Singer route no longer applies.  This algebraic distinction
changes the subsequent echo-state analysis: the composite construction must
instead be treated through its component hit words in
Sec.~\ref{sec:S_esp}.

There is also a simple bound on the period reached by pairwise-coprime Singer
components.  If block \(\mu\) contains \(n_\mu\) qubits, then
\(K_\mu=2^{n_\mu}+1\) and \(\sum_\mu n_\mu=N\).  Pairwise coprimality requires
the largest powers of two dividing the integers \(n_\mu\) to be distinct.
Indeed, writing \(n_\mu=2^{a_\mu}b_\mu\) with \(b_\mu\) odd, two periods
\(2^{n_\mu}+1\) and \(2^{n_\nu}+1\) share a nontrivial factor whenever
\(a_\mu=a_\nu\).  After ordering the distinct nonnegative integers \(a_\mu\),
we have \(a_\mu\ge\mu-1\) and hence \(n_\mu\ge2^{\mu-1}\).  If
\(n_{\rm blk}\) denotes the number of component reservoirs, then
\begin{equation}
    \frac{K_{\rm comp}}{2^N}
    =\prod_\mu\left(1+2^{-n_\mu}\right)
    \le
    \prod_{\mu=1}^{n_{\rm blk}}\left(1+2^{-2^{\mu-1}}\right)
    =2\left(1-2^{-2^{n_{\rm blk}}}\right)<2 .
    \label{eq:S_composite_period_bound}
\end{equation}
Hence \(K_{\rm comp}<2^{N+1}\).  For \(N=3\), the construction in
Eq.~\eqref{eq:S_explicit_V15} realizes \(K_{\rm comp}=15<16\).  Composite
routing can therefore extend a particular route beyond the single Singer
period while retaining the same \(O(2^N)\) scaling.

\subsection{Local echo-state conditions for Clifford routes}
\label{sec:S_esp}

Throughout this work, the echo-state property (ESP) refers to the local
property around the constant operating history \(s_t=s_0\): reservoir states
driven by the same sufficiently small input perturbations lose their dependence
on the initial state.  \textbf{For the channel considered here, this local ESP
is equivalent to every nonidentity Pauli orbit encountering at least one
damping hit.}  A transparent orbit would preserve a traceless component of the
initial state indefinitely.  Strict contraction is established below at
\(s_t=s_0\) and persists in a sufficiently small neighborhood by continuity.
No uniform claim is made for arbitrary input amplitudes; the finite-amplitude
sequences used in the task benchmark are tested directly in
Sec.~\ref{sec:S_finite_amplitude_washout}.  In this subsection,
\(\oplus\) denotes addition in \(\GF\), whereas \(\sum\) denotes an ordinary
integer sum that counts physical damping hits.

Consider a nonidentity Pauli orbit
\(\cO=\{P_1,\ldots,P_{K_{\cO}}\}\) under \(V\), and let
\[
    n_{\rm hit}(\cO)
    =\sum_{j=1}^{K_{\cO}}\eta_j
    \in\mathbb Z_{\ge0}
\]
be the number of damping hits in one period.  If \(\bm v_P\) and
\(\bm v_G\) are the binary labels of \(P_1\) and \(G\), respectively, then
\begin{equation}
    \eta_j=\symp{\bm v_G}{M^{j-1}\bm v_P}\in\GF,
    \label{eq:S_singer_hit_shadow}
\end{equation}
where \(\eta_j=0\) denotes commutation and \(\eta_j=1\) denotes
anticommutation.

\begin{theorem}[Local ESP and finite transparent runs]
For the Clifford write-store channel with \(0<r<1\), the operating-point
dynamics has local ESP iff
\begin{equation}
    n_{\rm hit}(\cO)>0
    \qquad
    \text{for every nonidentity Pauli orbit }\cO.
    \label{eq:S_finite_hitting}
\end{equation}
Whenever this condition holds, no hit word contains \(2N\) consecutive
zeros.  Equivalently, along every nonidentity orbit, any \(2N\) consecutive
routed Pauli strings contain at least one string that anticommutes with \(G\).
\end{theorem}

\begin{proof}
After one complete period,
\[
    \left(\Phi_{s_0,r}^{V,Q,G}\right)^{K_{\cO}}(P_1)
    =\pm r^{n_{\rm hit}(\cO)}P_1 .
\]
The sign is irrelevant for contraction.  The hitting condition places at
least one factor \(r<1\) on every nonstationary Pauli orbit, so the
operating-point channel has spectral radius strictly below one on the
traceless operator space.  Since the channel is unital,
\(\rho_\ast=\id_{\rm R}/d\) is then its unique attracting fixed point.  In
finite dimension, there is a norm in which this channel contracts every
traceless operator by a common factor \(c<1\).  Continuity in \(s\) preserves
a strict common contraction throughout some neighborhood of \(s_0\), so any
input sequence confined to this neighborhood exponentially washes out the
initial-state difference.  This proves the local ESP.  Conversely, local ESP
includes the constant sequence \(s_j=s_0\), for which an undamped Pauli orbit
would retain an initial-state component indefinitely; the hit condition is
therefore also necessary.

The second statement follows directly from the characteristic polynomial in
Eq.~\eqref{eq:S_characteristic_polynomial}.  Applying Cayley--Hamilton to
Eq.~\eqref{eq:S_singer_hit_shadow} gives the binary recurrence
\begin{equation}
    \eta_{j+2N}
    \oplus c_{2N-1}^{(M)}\eta_{j+2N-1}
    \oplus\cdots
    \oplus c_1^{(M)}\eta_{j+1}
    \oplus\eta_j
    =0
    \qquad\text{in }\GF .
    \label{eq:S_hit_recurrence}
\end{equation}
If \(2N\) consecutive entries vanished, this recurrence would propagate the
zeros forward.  Its constant coefficient is one, so it can also be solved
backward, forcing the entire orbit word to vanish.  That orbit would have
\(n_{\rm hit}(\cO)=0\), contradicting local ESP.  Therefore
\[
    \sum_{m=j}^{j+2N-1}\eta_m\ge1
    \qquad
    \text{for every }j .
\]
The sum in this last inequality is an ordinary count, not a modulo-two sum.
\end{proof}

\begin{corollary}[Automatic local ESP for Singer routing]
If the Clifford generator produces a Singer cycle, then any nonidentity
damping string \(G\) gives local ESP at the operating point for \(0<r<1\).
\end{corollary}

\begin{proof}
Equation~\eqref{eq:S_singer_orbit_basis} shows that any \(2N\) consecutive
labels on a nonzero Singer orbit form a basis.  If their hit values all
vanished, Eq.~\eqref{eq:S_singer_hit_shadow} would make \(\bm v_G\)
symplectically orthogonal to a basis.  Nondegeneracy of the symplectic form
would then force \(\bm v_G=0\), contrary to \(G\ne I\).  Thus every
nonidentity orbit contains a damping hit.
\end{proof}

This automatic conclusion is special to irreducible Singer routing.  For a
reducible composite generator, the local structure must also be used.

\begin{proposition}[Local ESP for coprime Singer components]
Assume that all component routes are Singer cycles with pairwise coprime
periods and that the restriction \(G_\mu\) of the damping string is
nonidentity on every component.  Then the product route satisfies
Eq.~\eqref{eq:S_finite_hitting} and therefore has local ESP for \(0<r<1\).
\end{proposition}

\begin{proof}
Pairwise coprimality implies that one global period visits every tuple of
local route positions.  Suppose that the global hit parity in
Eq.~\eqref{eq:S_product_eta} vanished for every such tuple.  Fixing all local
route positions except one and varying the remaining position would force the
corresponding active local hit word to be constant.  This cannot occur for a
local Singer cycle with \(G_\mu\ne I\).  Its local routing matrix \(M_\mu\)
has an irreducible characteristic polynomial.  An all-zero local word is
excluded by the Singer corollary, while an all-one local word would solve
Eq.~\eqref{eq:S_hit_recurrence} only if
\(\chi_{M_\mu}(1)=0\) in \(\GF\), which would make \(\lambda+1\) a factor of
\(\chi_{M_\mu}(\lambda)\).  Thus neither constant word is possible, and every
nonidentity product orbit contains a damping hit.  No irreducibility is
assumed for the global product matrix \(M_{\rm comp}\).
\end{proof}

The proposition is only a sufficient condition.  If a component damping
string is trivial, an orbit supported entirely on that component may remain
undamped; more generally, noncoprime component periods can correlate the
local phases and allow the parity word to vanish.  Such cases must be checked
orbit by orbit using the theorem above.  Even when ESP holds, that theorem
still requires a hit in every window of \(2N\) consecutive steps.
Composite routing can therefore lengthen the full period, but cannot create
an arbitrarily long transparent segment.

\subsection{Writing and resolving delayed products}
\label{sec:S_route_specific_second_order}

The preceding analysis concerns first-order Fisher memory, its stability,
and the attainable route length.  We finally turn to nonlinear tasks, which
depend on how an incoming perturbation acts on information already present
in the reservoir.  To keep the chain rule readable, this subsection uses
\(\Phi_s\equiv\Phi_{s,r}^{V,Q,G}\).  For
\(1\le k<j\le L_{\rm hist}\), the coordinate \(s_j\) is
older than \(s_k\).  To expose the chain rule, first differentiate with
respect to \(s_j\) while leaving the \(k\)th channel undifferentiated.  At
the operating point, all channels older than \(s_j\) act on the fixed
point and hence leave it unchanged, giving
\[
    \left.
    \frac{\partial\rho_f}{\partial s_j}
    \right|_{\substack{s_a=s_0\\ a\ne k}}
    =
    \Phi_{s_0}^{k-1}\!\left[
        \Phi_{s_k}\!\left(
            \Phi_{s_0}^{j-k-1}
            [\dot\Phi_{s_0}(\rho_\ast)]
        \right)
    \right].
\]
The older perturbation is therefore created by the rightmost
\(\dot\Phi_{s_0}\), propagated for \(j-k-1\) steps, acted on by the
derivative of the \(k\)th channel, and finally propagated for another
\(k-1\) steps to the readout.  Differentiating the preceding expression
with respect to \(s_k\) gives the mixed tangent
\begin{equation}
    \Gamma_{kj}
    =
    \left.
    \frac{\partial^2\rho_f}{\partial s_k\,\partial s_j}
    \right|_{\bm s=s_0\bm{1}}
    =
    \Phi_{s_0}^{k-1}\!\left[
        \dot\Phi_{s_0}\!\left(
            \Phi_{s_0}^{j-k-1}
            [\dot\Phi_{s_0}(\rho_\ast)]
        \right)
    \right].
    \label{eq:S_Gamma_def}
\end{equation}
Here \(\bm 1\) denotes the all-ones history vector, so
\(\bm s=s_0\bm1\) is the operating-point history.
No \(\ddot\Phi_{s_0}\) term appears because \(k\ne j\): the two
derivatives act on two distinct occurrences of the driven channel in
Eq.~\eqref{eq:S_history_state}.
Evaluating this derivative requires the action of the input-derivative
superoperator on an already stored Pauli string.  Applying
Eq.~\eqref{eq:S_many_dotPsi_general} to \(P\) and using \(\{Q,G\}=0\)
gives
\begin{equation}
    \dot\Psi_{s_0,r}^{Q,G}(P)
    =
    \begin{cases}
        \pi\sqrt{1-r^2}\,PQ,& [P,Q]=[P,G]=0,\\
        0,&\text{otherwise}.
    \end{cases}
    \label{eq:S_many_dotPsi}
\end{equation}
The older perturbation arrives at the second write event as
\begin{equation}
    \Phi_{s_0}^{j-k-1}\dot\Phi_{s_0}(\rho_\ast)
    =
    \frac{\pi\sqrt{1-r^2}}{d}
    r^{n_{j-k-1}(\eta)}P_{j-k}.
    \label{eq:S_old_tangent}
\end{equation}
It follows from Eq.~\eqref{eq:S_many_dotPsi} that the mixed response is
nonzero only if
\begin{equation}
    [P_{j-k},Q]=[P_{j-k},G]=0.
    \label{eq:S_second_condition}
\end{equation}
When this condition holds, define
\begin{equation}
    \widetilde P_1=V(P_{j-k}Q)V^\dagger,\qquad
    \widetilde P_{m+1}=V\widetilde P_mV^\dagger,
    \label{eq:S_product_route}
\end{equation}
and let \(\widetilde\eta_m=0\) or \(1\) according as
\(\widetilde P_m\) commutes or
anticommutes
with \(G\).  The resulting mixed tangent is
\begin{equation}
    \Gamma_{kj}
    =
    \frac{\pi^2(1-r^2)}{d}\,
    r^{
        n_{j-k-1}(\eta)
        +
        \sum_{m=1}^{k-1}\widetilde\eta_m
    }
    \widetilde P_k.
    \label{eq:S_Gamma_Pauli}
\end{equation}
Thus a delayed product is written along a definite Pauli route and can be
accessed by a linear readout of Pauli observables.

For the main-text target
\[
    y_t^{\rm prod}(q)=u_{t-q}u_{t-1},\qquad q\ge2,
\]
with the reservoir read after processing the current input, the relevant
indices are given below.  Throughout this discussion, \(q\) denotes only the
integer delay; binary Pauli labels are written as bold vectors.
\begin{equation}
    k=2,\qquad j=q+1,\qquad j-k=q-1.
    \label{eq:S_product_indices}
\end{equation}
If \(\bm v_Q,\bm v_G\in\GF^{2N}\) denote the binary labels of \(Q\) and
\(G\), define the delay-\(q\) write-in indicator by
\begin{equation}
    w_{\rm prod}(q;Q,G)=1
    \quad\Longleftrightarrow\quad
    \symp{\bm v_Q}{M^{q-1}\bm v_Q}=0
    \ \text{and}\
    \symp{\bm v_G}{M^{q-1}\bm v_Q}=0.
    \label{eq:S_product_write_indicator}
\end{equation}
Thus \(w_{\rm prod}=1\) means that the stored Pauli string commutes with
both \(Q\) and \(G\), so the leading mixed tangent exists.  If
\(w_{\rm prod}=0\), Eq.~\eqref{eq:S_many_dotPsi} gives
\(\Gamma_{2,q+1}=0\).  Only when \(w_{\rm prod}=1\) can the origin of the
response label be
followed directly through the last two write-store-routing steps.  The older
tangent reaches the \(u_{t-1}\) write event with label \(M^{q-1}\bm v_Q\).
Multiplication by the newly injected \(Q\) adds \(\bm v_Q\) modulo two.  The
routing operation within that same channel then contributes one factor of
\(M\), and the operating-point channel for the current input \(u_t\)
contributes a second factor before readout:
\[
    M^{q-1}\bm v_Q
    \ \xrightarrow{\ \times Q\ }\
    \bm v_Q\oplus M^{q-1}\bm v_Q
    \ \xrightarrow{\ V\text{ at }u_{t-1}\ }\
    M\!\left(\bm v_Q\oplus M^{q-1}\bm v_Q\right)
    \ \xrightarrow{\ V\text{ at }u_t\ }\
    M^2\!\left(\bm v_Q\oplus M^{q-1}\bm v_Q\right).
\]
Thus the binary Pauli label of the second-order response for delay \(q\) is
\begin{equation}
    \bm p_q^{(2)}
    =
    M^2\!\left[
        \bm v_Q\oplus\left(M^{q-1}\bm v_Q\right)
    \right]
    \in\GF^{2N}.
    \label{eq:S_product_label}
\end{equation}
Equation~\eqref{eq:S_product_label} determines the Pauli direction of the
second-order tangent only when \(w_{\rm prod}(q;Q,G)=1\).  If this write-in
condition fails, the tangent vanishes and the corresponding product is not
encoded at second order.  If it holds, the response is nonzero, but its label
\(\bm p_q^{(2)}\) may still coincide with a first-order memory direction, so
that the product signal shares an observable with a linear echo.  The analysis
below treats these two cases separately.

The aliasing condition has a particularly simple form on a genuine Singer
route.  Retaining the convention above, \(q\) is the integer delay in the
target \(u_{t-q}u_{t-1}\), while \(\ell=q-1\) is the corresponding separation
along the Pauli route.

\begin{proposition}[Singer-route aliasing count]
Let \(M\) generate a Singer route of period \(K=2^N+1\).  Conditional on a
nonzero write-in indicator \(w_{\rm prod}(q;Q,G)\), the second-order label
in Eq.~\eqref{eq:S_product_label} can coincide with a first-order route label
only when \(3\mid K\), and then only at
\begin{equation}
    q=1+\frac K3,
    \qquad
    q=1+\frac{2K}{3}.
    \label{eq:S_singer_collision_delays}
\end{equation}
The full return \(q=K+1\) is a separate zero of the mixed response.  Since
\(3\mid(2^N+1)\) exactly when \(N\) is odd, no such route aliases occur for
even \(N\); for odd \(N\), there are at most the two candidates above.  Thus
the number of possible aliases remains bounded by two as \(N\) grows, while
the return zero is always present.
\end{proposition}

This proposition classifies aliases only after a nonzero mixed response has
been written.  Whether that response is present still depends on the
\(Q,G\)-dependent write-in condition \(w_{\rm prod}\), which can introduce
additional zeros.

\begin{proof}
Let \(\ell=q-1\) and write \(\bm v=\bm v_Q\).  Aliasing means that the
binary sum in Eq.~\eqref{eq:S_product_label} lies on the first-order route.
The common prefactor \(M^2\) merely shifts the route index, so for some
integer \(m\), defined modulo \(K\),
\begin{equation}
    \bm v\oplus M^\ell\bm v=M^m\bm v .
    \label{eq:S_singer_collision_vector}
\end{equation}
For a Singer route, the \(2N\) vectors
\(\bm v,M\bm v,\ldots,M^{2N-1}\bm v\) form a basis of \(\GF^{2N}\), as
shown in Eq.~\eqref{eq:S_singer_orbit_basis}.  Let
\(P(M)=\id_{2N}\oplus M^\ell\oplus M^m\).  Since \(P(M)\) is a polynomial in
\(M\), it commutes with \(M\).  Thus \(P(M)\bm v=0\) implies
\(P(M)M^a\bm v=M^aP(M)\bm v=0\) for every vector in the orbit basis.  The
matrix \(P(M)\) therefore vanishes on the entire label space, and hence
\begin{equation}
    \id_{2N}\oplus M^\ell=M^m,
    \label{eq:S_singer_collision_matrix}
\end{equation}
where \(\oplus\) denotes matrix addition modulo two.

We next extract the allowed exponents from this matrix identity.  The only
unusual step is that all matrix entries and additions are in \(\GF\).  Since
\(\id_{2N}\) and \(M^\ell\) commute, squaring their binary sum gives
\[
    \left(\id_{2N}\oplus M^\ell\right)^2
    =\id_{2N}\oplus M^\ell\oplus M^\ell\oplus M^{2\ell}
    =\id_{2N}\oplus M^{2\ell}.
\]
The two middle terms cancel because a matrix added to itself is zero over
\(\GF\).  Squaring both sides of
Eq.~\eqref{eq:S_singer_collision_matrix} repeatedly therefore gives, after
\(a\) squarings,
\[
    \id_{2N}\oplus M^{2^a\ell}=M^{2^a m}.
\]
Set \(a=N\).  Since \(K=2^N+1\), we have
\(2^N\equiv-1\pmod K\).  Moreover, \(M^K=\id_{2N}\), so exponents of
\(M\) may be reduced modulo \(K\).  The repeatedly squared identity is thus
\[
    \id_{2N}\oplus M^{-\ell}=M^{-m}.
\]
We now obtain a second expression with exactly the same left-hand side.
Multiplying the original identity
\(\id_{2N}\oplus M^\ell=M^m\) by \(M^{-\ell}\) gives
\[
    \id_{2N}\oplus M^{-\ell}=M^{m-\ell}.
\]
Equating the two right-hand sides and using the exact order \(K\) of \(M\)
yields
\begin{equation}
    M^{-m}=M^{m-\ell}
    \quad\Longrightarrow\quad
    -m\equiv m-\ell\pmod K
    \quad\Longrightarrow\quad
    \ell\equiv2m\pmod K .
    \label{eq:S_singer_collision_congruence}
\end{equation}
Here exact order means that \(M^a=M^b\) holds precisely when
\(a\equiv b\pmod K\).

It remains to determine which values of \(m\) can satisfy this congruence.
Let \(X=M^m\).  Substituting \(\ell\equiv2m\) into
Eq.~\eqref{eq:S_singer_collision_matrix} gives
\[
    \id_{2N}\oplus X\oplus X^2=0.
\]
Multiplication by \(\id_{2N}\oplus X\) is useful because all intermediate
terms cancel in pairs over \(\GF\):
\[
    0
    =\left(\id_{2N}\oplus X\right)
     \left(\id_{2N}\oplus X\oplus X^2\right)
    =\id_{2N}\oplus X^3.
\]
Consequently \(X^3=\id_{2N}\), or equivalently
\begin{equation}
    M^{3m}=\id_{2N}.
    \label{eq:S_singer_collision_order_three}
\end{equation}
Because \(M\) has exact order \(K\), Eq.~\eqref{eq:S_singer_collision_order_three}
is equivalent to \(3m\equiv0\pmod K\).  The solution \(m\equiv0\) is
impossible in Eq.~\eqref{eq:S_singer_collision_matrix}, because it would
require the invertible matrix \(M^\ell\) to vanish.  Nonzero solutions
therefore exist only if \(3\mid K\), in which case
\(m\equiv K/3\) or \(2K/3\pmod K\).  Equation~\eqref{eq:S_singer_collision_congruence}
then gives \(\ell\equiv2K/3\) or \(K/3\pmod K\), respectively.  Since
\(q=\ell+1\), these are precisely the two delays in
Eq.~\eqref{eq:S_singer_collision_delays}.  Finally,
\(K=2^N+1\) is divisible by three exactly when \(N\) is odd, because
\(2\equiv-1\pmod3\).  The full return \(\ell=K\) is different:
\(M^K\bm v=\bm v\), so \(\bm v\oplus M^K\bm v=0\) and the mixed label
vanishes rather than aliasing a nonzero first-order direction.
\end{proof}

The implication used in the proof is therefore one-way.  If
\(w_{\rm prod}(q;Q,G)=1\) and the resulting nonzero second-order label
aliases a first-order route direction, then \(q\) must be one of the two delays
in Eq.~\eqref{eq:S_singer_collision_delays}.  The converse need not hold:
the write-in condition may still give \(w_{\rm prod}=0\), and at all other
delays the mixed response may likewise be absent for the chosen \(Q,G\).

\begin{proposition}[Initial second-order write-in window]
For any nonidentity write string \(Q\) on an \(N\)-qubit Singer route, let
\(L_2\) be the largest integer such that
\(w_{\rm prod}(\ell+1;Q,G)=1\) for every
\(\ell=1,\ldots,L_2\).  Then
\begin{equation}
    L_2\le 2N-2.
    \label{eq:S_product_initial_run_bound}
\end{equation}
Thus a continuously supported initial product window grows at most linearly
with the reservoir size, or logarithmically with the Singer period
\(K=2^N+1\).
\end{proposition}

\begin{proof}
The first condition in
Eq.~\eqref{eq:S_product_write_indicator} gives
\(\symp{\bm v_Q}{M^\ell\bm v_Q}=0\) for
\(\ell=1,\ldots,L_2\); the same expression also vanishes at \(\ell=0\)
because the symplectic form is alternating.  If \(L_2\ge2N-1\), the nonzero
label \(\bm v_Q\) would therefore be symplectically orthogonal to
\[
    \bm v_Q,M\bm v_Q,\ldots,M^{2N-1}\bm v_Q .
\]
These \(2N\) consecutive Singer labels form a basis of \(\GF^{2N}\) by
Eq.~\eqref{eq:S_singer_orbit_basis}.  Nondegeneracy of the symplectic form
would then force \(\bm v_Q=0\), contradicting \(Q\ne I\).
\end{proof}

The \(K=9\) model gives the smallest example of the two universal Singer
aliases.  Its routing matrix obeys the following identity over \(\GF\):
\begin{equation}
    M_9^6+M_9^3+\id_6=0 .
    \label{eq:S_K9_order_three}
\end{equation}
Here the additions are modulo two.  Applying this identity to \(\bm v_Q\)
gives
\[
    \bm v_Q\oplus\left(M_9^3\bm v_Q\right)=M_9^6\bm v_Q,
    \qquad
    \bm v_Q\oplus\left(M_9^6\bm v_Q\right)=M_9^3\bm v_Q.
\]
Equation~\eqref{eq:S_product_label} therefore reduces to
\begin{equation}
    \bm p_4^{(2)}
    =M_9^2\!\left[
        \bm v_Q\oplus\left(M_9^3\bm v_Q\right)
    \right]
    =M_9^8\bm v_Q,
    \qquad
    \bm p_7^{(2)}
    =M_9^2\!\left[
        \bm v_Q\oplus\left(M_9^6\bm v_Q\right)
    \right]
    =M_9^5\bm v_Q.
    \label{eq:S_K9_product_collisions}
\end{equation}
Whenever the write-in condition in
Eq.~\eqref{eq:S_product_write_indicator} holds, the \(q=4\) and \(q=7\)
product tangents therefore alias the first-order
route at \(P_8\) and \(P_5\), respectively.  For the fixed benchmark pair
\(Q=XIX\), \(G=ZXI\), however, the write-in indicator vanishes at both
delays.  Their leading mixed tangents are therefore absent rather than
merged with these linear-memory directions.

The preceding \(K=9\) example concerns a single irreducible Singer route.
The second fixed Clifford reservoir in Fig.~4 instead uses the composite
period-15 route introduced in Eq.~\eqref{eq:S_explicit_V15}.  Its product
response can fail for an additional, physically transparent reason: one
component route may return before the complete product route closes.  To make
these partial returns explicit, write
\begin{equation}
    M_{15}=M_5\oplus M_3,
    \qquad
    \bm v_Q=(\bm v_{Q,5},\bm v_{Q,3}).
    \label{eq:S_K15_product_clock_labels}
\end{equation}
Here \(\oplus\) denotes the direct sum of the two component routing matrices,
not the binary addition used between Pauli labels.  The response label then
separates into the corresponding components:
\begin{equation}
    \bm p_q^{(2)}
    =
    \left(
        M_5^2\!\left[
            \bm v_{Q,5}\oplus\left(M_5^{q-1}\bm v_{Q,5}\right)
        \right],
        M_3^2\!\left[
            \bm v_{Q,3}\oplus\left(M_3^{q-1}\bm v_{Q,3}\right)
        \right]
    \right).
    \label{eq:S_K15_product_response}
\end{equation}
For \(q=6\) and \(q=11\), the exponent \(q-1\) is a nonzero multiple of
\(5\).  The two-qubit component has returned,
\(M_5^{q-1}\bm v_{Q,5}=\bm v_{Q,5}\), whereas the one-qubit component has
not.  The stored string then fails the first condition in
Eq.~\eqref{eq:S_second_condition}, independently of \(G\), and the leading
mixed tangent vanishes.  At \(q=16\), the full route returns to \(Q\), and
the second condition fails because \(\{Q,G\}=0\).  These obstructions account
for the corresponding peaks in the numerical product benchmark.

\subsection{A parameter-counting estimate beyond Clifford routing}

Equation~\eqref{eq:S_composite_period_bound} shows that coprime product routes
extend the available period without changing its \(O(2^N)\) scale.  This
raises a broader question.  Is the same scale expected when the recurrent
dynamics is not restricted to Clifford routing?  The following argument is a
local parameter-counting estimate for regular solution families, not a
universal upper bound.  Let \(d=2^N\).  A pure input qubit and a joint unitary
\(U\in U(2d)\) generate a channel of Kraus rank at most two.  After removing a
unitary rotation of the discarded output qubit and a common change of
reservoir basis, the effective dimension of a generic rank-two family is
\begin{equation}
    n_{\rm par}(d)=4d^2-4-(d^2-1)=3d^2-3
    \label{eq:S_rank2_parameter_dimension}
\end{equation}
independent real parameters.

Let \(\widehat\Delta_j\) denote the normalized tangent obtained after \(j\)
reservoir steps.  A length-\(K\) cycle requires \(K(K-1)/2\) pairwise
orthogonality conditions.  Closing the normalized route,
\(\widehat\Delta_K=\pm\widehat\Delta_0\), imposes another
\(d^2-2\) conditions because the traceless Hermitian operator space has
dimension \(d^2-1\).  The total number of independent equalities is therefore
\begin{equation}
    n_{\rm eq}(K)=\frac{K(K-1)}{2}+d^2-2.
    \label{eq:S_cycle_constraint_count}
\end{equation}
If the orthogonality and closure constraints have independent gradients at a
solution, which is the regularity assumption used here, they cannot outnumber
the available parameters.  Thus \(n_{\rm eq}(K)\le n_{\rm par}(d)\) gives
\begin{equation}
    K\le 2d=2^{N+1}.
    \label{eq:S_regular_cycle_bound}
\end{equation}
Combining this estimate with the dimension of the traceless operator space
gives
\[
    K\le\min\{2d,d^2-1\}.
\]
This estimate does not exclude isolated constructions protected by additional
symmetries.  ESP does
not change the count because strict contraction is an inequality rather than
an additional equality.  For the three-qubit reservoirs in the main text,
\(d=8\) and Eq.~\eqref{eq:S_regular_cycle_bound} gives \(K\le16\); the
period-\(15\) product route is therefore close to the natural
single-input-qubit scale.

\section{Numerical Protocols for the Letter}
\label{sec:S_numerical_protocols}

This section specifies the simulations behind Figs.~1, 3, and 4 of the
Letter.  The first two subsections test the single-qubit optimum and enumerate
Clifford damping patterns.  The last subsection gives the common data,
measurement, and regression protocol used to compare the three fixed
reservoirs in the task benchmark.

\subsection{Optimization over unrestricted
\texorpdfstring{\(SU(4)\)}{SU(4)} dynamics}

The red points in Fig.~1 test whether a generic input--reservoir interaction
can improve upon the Fisher cycle.  We directly minimized
\(C_{\rm GM}(\Omega_\alpha,\mathsf H)\) over \(U\in SU(4)\) at
\(s_0=1/2\), using
\[
    \Omega_\alpha
    =
    \operatorname{diag}\!\left(
        \frac{1-\alpha}{2},
        \frac{1-\alpha}{2},
        \alpha
    \right),
    \qquad
    \alpha\in\{0,1/50,\ldots,1\}\cup\{1/3\}.
\]
For every trial gate, the
channel fixed point and the three-delay QFIM were recomputed.  The channel
derivative was evaluated by a centered finite difference with step
\(10^{-6}\).

At each \(\alpha\), the four lowest-cost gates from an independent scan of
\(2\times10^4\) Haar-random \(SU(4)\) matrices were used as initial points.
The initial pool was screened to ensure that each gate defines a valid
fading-memory model at the operating point.  The computed fixed point was
required to satisfy the channel equation to residual \(10^{-9}\), while the
spectral radius of the nonstationary sector was required to be below
\(1-10^{-8}\).  The latter condition enforces strict contraction toward the
fixed point and is the numerical ESP criterion used here.  We further
required the three-delay QFIM to be nonsingular
\((\lambda_{\min}>10^{-9})\), so that all three local input directions are
identifiable and the GM cost is finite.

Each seed was refined for \(420\) trials.  A trial direction was generated
from a complex Ginibre matrix by taking its Hermitian part, removing its
trace, and normalizing its Frobenius norm.  The resulting \(\Xi\) explores all
\(15\) independent generator directions of \(SU(4)\), so the search is not
restricted to Clifford gates or to the write--store--routing ansatz.  The
trial gate was \(U'=e^{-i\epsilon\Xi}U\).  Cost-reducing moves were accepted
directly, while uphill moves were accepted with probability
\(\exp[(C_{\rm GM}^{\rm old}-C_{\rm GM}^{\rm new})/T_{\rm SA}]\), where the
simulated-annealing temperature is
\(T_{\rm SA}=\max[10^{-5},0.015\epsilon]\).  The step began at
\(\epsilon=0.16\), was multiplied by \(0.986\) after every trial, and was
bounded below by \(10^{-5}\).  The best gate among the four refined seeds
gives the red point.  Within this search, no optimized \(SU(4)\) gate
improved upon the Fisher-cycle frontier.

\subsection{Enumeration and optimization of damping patterns}

For the three-qubit Clifford scans, we enumerated admissible Pauli pairs
\((Q,G)\) satisfying \(\{Q,G\}=0\), the target route condition, and the global
finite Pauli hitting condition.  For each admissible pair, the hit word
\(\eta_K(V,Q,G)\) was computed from Eq.~\eqref{eq:S_eta_word}, and the QFIM was
computed analytically from Eq.~\eqref{eq:S_many_H_eta}.  The uniform-window
score used in the main text is
\begin{equation}
    C_{\rm GM}^{\rm unif}(\eta_K)
    =
    \min_{0<r<1}
    \frac{1}{K}
    \left[
        \sum_{k=1}^K
        \frac{1}{
            \sqrt{
            \left(\mathsf H^{(K)}_{\eta_K}(r)\right)_{kk}
            }
        }
    \right]^2,
    \label{eq:S_pattern_cost}
\end{equation}
where the common Hilbert-space dimension factor in the GM bound is omitted
because all three-qubit patterns being compared have the same dimension.  The
enumeration gives 32 admissible damping words for the \(K=9\) Singer route
and 22 admissible damping words for the \(K=15\) product route.

\subsection{Fixed-reservoir task benchmark}
\label{sec:S_dynamical_benchmarks}

The main-text benchmark evolves the full finite-amplitude channel rather
than the linearized tangent equations.  The input sequence is
\begin{equation}
    s_t=s_0+\delta w\,u_t,
    \qquad
    s_0=\frac12,
    \qquad
    u_t\sim{\rm Unif}[-1,1],
    \label{eq:S_task_input}
\end{equation}
with \(\delta w=0.25\) in the reported data.  The tested targets are
\begin{equation}
    y_t^{\rm lin}(q)=u_{t-q},
    \qquad
    y_t^{\rm prod}(q)=u_{t-q}u_{t-1},
    \qquad q=2,\ldots,20 .
    \label{eq:S_task_targets}
\end{equation}
The product target couples the input \(q\) steps before readout to the most
recent previous input, in agreement with the convention used in the
second-order analysis above.

\paragraph{Reservoir selection and data.}
One representative of each model class is fixed before the test sequences
are generated and is then used for every delay and both targets.  The two
Clifford reservoirs are
\begin{equation}
    (K,Q,G,r)=(9,XIX,ZXI,0.35),
    \qquad
    (K,Q,G,r)=(15,XIX,YYX,0.56).
    \label{eq:S_fixed_clifford_models}
\end{equation}
The Pauli pairs are fixed before the task simulations.  The fading parameter
of each route is then selected once on validation data using the randomized
local-Pauli readout at \(\delta w=0.25\) and \(N_{\rm shot}=6400\).  We scan
\(r=0.28,0.29,\ldots,0.40\) for \(K=9\) and
\(r=0.51,0.52,\ldots,0.64\) for \(K=15\), minimizing the equally weighted
mean of the linear and product validation NRMSE over \(q=2,\ldots,8\).
This gives the values in Eq.~\eqref{eq:S_fixed_clifford_models}, which are
then held fixed for all reported delays, both tasks, and both readout
conditions.  The tuning sequences are disjoint from those used for the
reported curves.
For the task-independent screening below, each Ising candidate is evaluated
through its \(K=2^3+1=9\)-delay QFIM.  Here
\(\Tr\mathsf H\) is the total delay sensitivity,
\(K\lambda_{\min}(\mathsf H)\) measures the weakest resolved input
combination after removing the trivial window-length scale,
\(r_{\rm eff}=(\Tr\mathsf H)^2/\Tr(\mathsf H^2)\) estimates the number of
resolved delay directions, and
\(\operatorname{cond}(\mathsf H)=\lambda_{\max}/\lambda_{\min}\) measures
their imbalance.  Large values of the first three diagnostics and a small
condition number are favorable; Sec.~\ref{sec:S_qfim_scaling} discusses them
in detail.

The Ising baseline is selected once from an independently generated library
of \(1000\) reservoirs satisfying the operating-point ESP test described
below.  Evaluating the full task benchmark for every library member would be
unnecessarily expensive, so we first retain candidates with simultaneously
large total Fisher information, large weakest-delay information, broad usable
rank, and moderate spectral imbalance.  For a metric \(x_i\) of candidate
\(i\), let \(\operatorname{pos}_{\downarrow}(x_i)\) denote its ordinal position
after sorting the \(1000\) values from largest to smallest, and let
\(\operatorname{pos}_{\uparrow}(x_i)\) denote the position after sorting from
smallest to largest; position one is best in either ordering.
The prescreening score is
\begin{equation}
\begin{aligned}
S_{{\rm pre},i}={}&
\operatorname{pos}_{\downarrow}\!\left(K\lambda_{\min}(\mathsf H_i)\right)
+\operatorname{pos}_{\downarrow}\!\left(r_{\rm eff}(\mathsf H_i)/K\right)\\
&+\operatorname{pos}_{\downarrow}\!\left(\Tr\mathsf H_i\right)
+\operatorname{pos}_{\uparrow}\!\left(\operatorname{cond}(\mathsf H_i)\right),
\end{aligned}
\label{eq:S_ising_prescreen_score}
\end{equation}
where \(\mathsf H_i\) is the delay-space QFIM of candidate \(i\).
The four ordinal ranks are summed so that no arbitrary relative units are
assigned to the four diagnostics.  We retain the \(120\) candidates with the
smallest \(S_{{\rm pre},i}\).  These candidates are then ranked by the
mean exact-readout validation NRMSE of the linear and product targets over
\(q=2,\ldots,8\) and three independent input sequences.  The selected Ising
unitary is held fixed thereafter; no test data and no finite-shot realization
enter this selection.

Each trajectory starts from the maximally mixed reservoir state
\(\rho_0=\id_{\rm R}/8\) and undergoes \(400\) washout steps.  We subsequently
collect \(1800\), \(900\), and \(900\) states for
training, validation, and test, respectively.  A linear ridge readout is
fitted for every target.  The Pauli features are standardized using the
training segment, and the ridge penalty is selected on the validation segment
from \(0\) and \(10^{-10},10^{-9},\ldots,10^{-2}\).  After this selection,
the feature normalization and readout coefficients are refitted on the
combined training and validation segments and evaluated once on the test
segment.  Performance is reported as the held-out normalized root-mean-square
error
\begin{equation}
    {\rm NRMSE}
    =
    \sqrt{
        \frac{
            \langle(\hat y_t-y_t)^2\rangle_{\rm test}
        }{
            \operatorname{Var}_{\rm test}(y_t)
        }
    }.
    \label{eq:S_NRMSE}
\end{equation}
The exact-readout curves are averaged over four independent input
realizations.  \textbf{Only the classical linear readout is refitted for each target;
the three reservoir dynamics remain fixed across all delays and both tasks.}

\paragraph{Exact Pauli features and shadow readout.}
Exact readout uses the expectation values of all \(4^3-1=63\) nonidentity
reservoir Pauli strings.  The finite-copy experiment instead samples one
uniformly random local Pauli basis on each copy.  Let
\(\bm b=(b_1,b_2,b_3)\in\{X,Y,Z\}^3\) denote the three selected bases and
\(\bm o=(o_1,o_2,o_3)\in\{+1,-1\}^3\) the corresponding outcomes.  For fixed
\((\bm b,\bm o)\), the product projector is
\(\bigotimes_a(\id+o_ab_a)/2\).  Each of the \(3^3=27\) settings is selected
with probability \(1/27\), so including the random setting in the recorded
outcome gives the \(27\times8=216\)-element POVM
\begin{equation}
    E_{\bm b,\bm o}
    =
    \frac{1}{27}
    \bigotimes_{a=1}^3\frac{\id+o_a b_a}{2}.
    \label{eq:S_shadow_povm}
\end{equation}
For a Pauli string \(P=\bigotimes_aP_a\), let \(\operatorname{wt}(P)\) be its number of
nonidentity factors.  A shot contributes only when every nonidentity \(P_a\)
matches the sampled basis \(b_a\), an event of probability
\(3^{-\operatorname{wt}(P)}\).
Multiplying the observed eigenvalues and compensating this matching
probability gives the classical-shadow estimator
\cite{SMHuang2020ClassicalShadows}
\begin{equation}
    \widehat p_P(\bm b,\bm o)
    =
    3^{\operatorname{wt}(P)}
    \prod_{a:P_a\ne I}o_a\,
    \mathbf 1[b_a=P_a],
    \qquad
    \mathbb E\widehat p_P=\Tr(\rho P).
    \label{eq:S_shadow_estimator}
\end{equation}
For example,
\(\widehat p_{XIZ}=9o_1o_3\mathbf 1[b_1=X]\mathbf 1[b_3=Z]\): only one ninth
of the settings match its two nonidentity factors, and the prefactor restores
the correct average.  Thus the same shadow record provides unbiased estimates
of all Pauli features, although higher-weight strings have fewer matching
shots and generally larger variances.  \textbf{The shadow protocol therefore
changes the precision with which these Pauli features are read, but not the
reservoir dynamics that created them.}
At every time step, the complete \(216\)-event distribution is sampled
multinomially with \(N_{\rm shot}=6400\).  All reconstructed Pauli features
therefore originate from the same outcome counts, retaining their sampling
correlations.  We use four independently generated input sequences and draw
four conditionally independent shadow records for each sequence.  The four
scores are first averaged at fixed input sequence, and the plotted mean and
standard error are then evaluated across the four sequence-level averages.
This preserves the nested sampling structure instead of treating repeated
shadow draws from the same trajectory as independent input realizations.
The same protocol is used for every curve in the right column of Fig.~4 of
the main text.

\paragraph{Clifford reservoirs.}
For the Clifford reservoirs, the full finite-amplitude channel is iterated
directly.  Given a reservoir Clifford \(V\), writing string \(Q\), damping
string \(G\), and fading parameter \(r\), the one-step unitary is
\begin{equation}
    U_{r,V,Q,G}
    =
    (\id_{\rm in}\otimes V)\,
    W_r(Q,G)\,
    (L_{\rm in}\otimes\id_{\rm R}),
    \qquad
    L_{\rm in}=e^{+\ii\pi Y/4},
    \label{eq:S_numeric_Ur}
\end{equation}
with
\begin{equation}
    W_r(Q,G)
    =
    \begin{pmatrix}
        a_r\id_{\rm R} & b_rQ\\
        b_rG & -a_rGQ
    \end{pmatrix},
    \qquad
    a_r=\sqrt{\frac{1+r}{2}},
    \qquad
    b_r=\sqrt{\frac{1-r}{2}} .
    \label{eq:S_numeric_Wr}
\end{equation}
The \(K=9\) model uses the Singer generator \(V_9\), whereas the \(K=15\)
model uses the product generator \(V_{15}=V_5\otimes V_3\).  Their complete
finite-amplitude unitaries are fixed by Eq.~\eqref{eq:S_fixed_clifford_models}
and are not reoptimized for individual targets.  No explicit history cutoff
is imposed; all earlier inputs remain in the recurrent state until they are
attenuated by the channel.

\paragraph{Random Ising baseline.}
The Ising library consists of fully connected Ising reservoirs with uniform
\(X\) and \(Z\) fields on one input qubit and three reservoir qubits.  Each joint unitary
has the form
\begin{equation}
    U_{\rm Ising}=e^{-\ii\tau_{\rm Ising}H_{\rm Ising}},
    \label{eq:S_ising_U}
\end{equation}
where
\begin{equation}
    H_{\rm Ising}
    =
    \sum_{a<b}J_{ab}X_aX_b
    +
    h^x\sum_a X_a
    +
    h^z\sum_a Z_a .
    \label{eq:S_ising_H}
\end{equation}
The indices \(a,b\) run over all four qubits before the input is traced out;
qubit \(1\) is the input and qubits \(2,3,4\) form the reservoir.  The fixed
Ising model used in Fig.~4 is specified by
\begin{equation}
\begin{gathered}
    (\tau_{\rm Ising},h^x,h^z)=(1.323,0.2193,-3.347),\\
    (J_{12},J_{13},J_{14},J_{23},J_{24},J_{34})
    =(-0.2175,-0.04014,-0.9958,0.8867,0.5425,0.8112).
\end{gathered}
    \label{eq:S_fixed_ising_parameters}
\end{equation}
No additional normalization is applied to these six couplings.
The parameters are drawn independently and uniformly from
\(J_{ab}\in[-1,1]\), \(h^x,h^z\in[-4,4]\), and
\(\tau_{\rm Ising}\in[0.05,10]\).  At \(s_0=1/2\), a sample is retained only
when its fixed-point residual is below \(10^{-9}\), its fixed state is full
rank to numerical tolerance, and the largest eigenvalue modulus on the
nonstationary channel sector is below \(1-10^{-5}\).  These tests remove
nonconvergent and nearly nonmixing dynamics before the task-based selection
described above.  The selected Ising unitary is then simulated with exactly
the same washout, trajectories, readout features, multinomial sampling, and
ridge protocol as the two Clifford reservoirs.

The systematic sharp peaks of the Clifford product curves coincide with the
route-specific response zeros and first/second-order aliases derived in
Sec.~\ref{sec:S_route_specific_second_order}.  Residual deviations can also
arise from finite input amplitude, shadow noise, and regression conditioning.

\section{Finite-Amplitude Robustness and Scaling with Reservoir Size}
\label{sec:S_robustness_scaling}

The preceding sections characterize the reservoir through its local tangent
dynamics and QFIM at the operating point \(s_0\).  Connecting these results to a
finite-data QRC requires three further checks.  The driven channel must erase
its initial condition at the finite input amplitudes used in the benchmark,
the encoded information must remain accessible with a finite measurement
budget, and the favorable QFIM structure must persist as the reservoir and its
memory window grow.  This section addresses these questions in turn.  We
verify finite-amplitude washout, derive and test the finite-copy behavior of
random local-Pauli shadows, and examine how the QFIM spectrum and second-order
write-in support vary with reservoir size.  These results delimit the regime
in which the local analytic conclusions apply to the task benchmarks in the
Letter.

\subsection{Finite-amplitude washout}
\label{sec:S_finite_amplitude_washout}

The algebraic result in Sec.~\ref{sec:S_esp} establishes local ESP near the
operating point but does not assert uniform contraction for arbitrary input
amplitudes.  We therefore test directly whether the finite-amplitude driven
dynamics used in the benchmark erase their initial condition.  This numerical
test supports ESP only over the sampled amplitudes and input families.  For a
common input sequence and two initial reservoir states, define
\begin{equation}
    D_t(\rho,\sigma)
    =
    \frac12\left\|
        \Phi_{s_t}\circ\cdots\circ\Phi_{s_1}(\rho-\sigma)
    \right\|_1 .
    \label{eq:S_finite_washout_distance}
\end{equation}
For every nonidentity three-qubit Pauli string \(P\), the Pauli probe is the
physical state pair
\(\rho_\pm^{(P)}=(\id\pm P)/8\).  Its initial difference is \(P/4\), so
\(D_0(\rho_+^{(P)},\rho_-^{(P)})=1\).  We propagate these \(63\) normalized
Pauli probes together with \(16\) random pairs of orthogonal pure states,
which have the same initial trace distance.  The input set contains constant,
alternating, and \(64\) independent random sequences at each of
\(\delta w=0,0.05,0.10,0.25,0.40\).  Figure~\ref{fig:S_finite_washout}
shows the maximum over all tested initial differences and sequences.

\begin{figure}[H]
    \centering
    \includegraphics[width=0.98\linewidth]{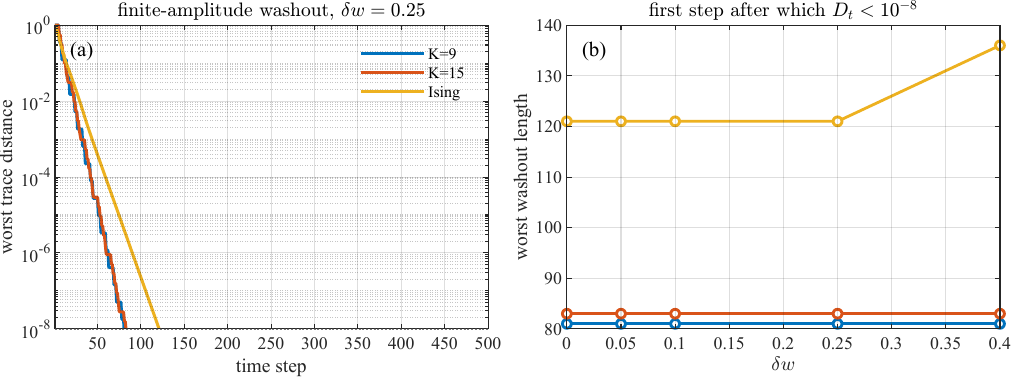}
    \caption{Finite-amplitude loss of the initial reservoir state.  (a) The
    worst trace distance at \(\delta w=0.25\).  (b) The first time after
    which the distance remains below \(10^{-8}\), maximized over the tested
    sequences.  All three reservoirs converge for every tested amplitude;
    the largest observed washout lengths are \(81\), \(83\), and \(136\) for
    \(K=9\), \(K=15\), and Ising, respectively.}
    \label{fig:S_finite_washout}
\end{figure}

At \(500\) steps the worst distances are between \(10^{-16}\) and
\(4\times10^{-15}\), close to numerical precision.  This finite test supports
the washout length used in the task simulations.  It is not presented as a
uniform proof over all bounded input sequences; rather, it verifies the
finite-amplitude regime and the three fixed dynamics used in Fig.~4.

\subsection{Classical Fisher information of the shadow readout}
\label{sec:S_shadow_fim}

The QFIM describes the information present in the state before a measurement
is chosen.  To determine how much of this information is accessible to the
fixed shadow protocol, consider the local model
\(\rho(\bm\theta)=\rho_\ast+\sum_k\theta_k\Delta_k+O(\|\bm\theta\|^2)\), where
\(\theta_k\) is the perturbation of the \(k\)th historical input.  The outcome
probability and its derivative at the operating point are
\[
    p_{\bm b,\bm o}
    =\Tr(\rho_\ast E_{\bm b,\bm o}),
    \qquad
    \partial_k p_{\bm b,\bm o}
    =\Tr(\Delta_k E_{\bm b,\bm o}).
\]
Substitution into the standard classical Fisher-information formula gives
\begin{equation}
    \mathsf I^{\rm sh}_{kl}
    =
    \sum_{\bm b,\bm o}
    \frac{
        \Tr(\Delta_k E_{\bm b,\bm o})
        \Tr(\Delta_l E_{\bm b,\bm o})
    }{
        \Tr(\rho_\ast E_{\bm b,\bm o})
    }.
    \label{eq:S_shadow_fim}
\end{equation}
Thus \(\mathsf H\) quantifies the information available in the quantum state,
whereas \(\mathsf I^{\rm sh}\) quantifies the part retained after the random
local-Pauli measurement is fixed; in general
\(\mathsf I^{\rm sh}\preceq\mathsf H\).  Equation~\eqref{eq:S_shadow_estimator}
and Eq.~\eqref{eq:S_shadow_fim} describe two uses of the same outcomes: the
former reconstructs Pauli features for regression, while the latter measures
their local sensitivity to the historical inputs.
For \(N_{\rm shot}\) independent copies, the total CFIM is
\(N_{\rm shot}\mathsf I^{\rm sh}\).  To compare only the coupling between delay
directions, Fig.~\ref{fig:S_fisher_matrices} displays
\begin{equation}
    \widetilde{\mathsf F}_{kl}
    =
    \frac{\mathsf F_{kl}}{\sqrt{\mathsf F_{kk}\mathsf F_{ll}}},
    \qquad \mathsf F\in\{\mathsf H,\mathsf I^{\rm sh}\}.
    \label{eq:S_normalized_fisher}
\end{equation}

\begin{figure}[H]
    \centering
    \includegraphics[width=0.98\linewidth]{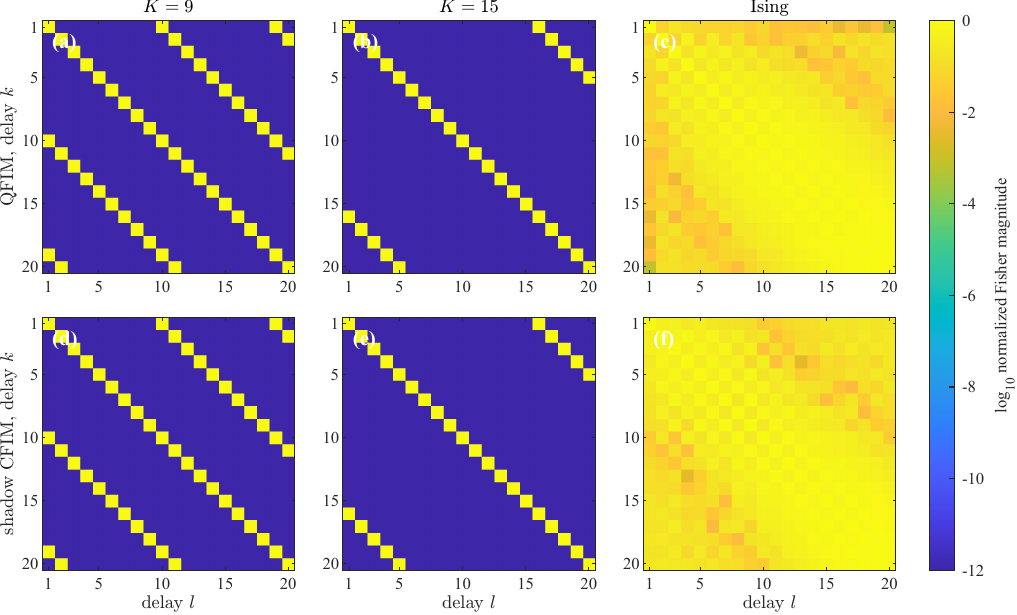}
    \caption{Delay-space Fisher structure of the three fixed reservoirs at
    \(s_0=1/2\).  The top row shows the normalized QFIM and the bottom row
    shows the normalized CFIM of the uniformly randomized
    local-Pauli POVM.  The Clifford routes retain their periodic diagonal
    structure after the measurement is fixed, whereas the Ising response is
    distributed over a dense set of delay directions.}
    \label{fig:S_fisher_matrices}
\end{figure}

The similarity between the two rows for the Clifford reservoirs does not
mean that the shadow POVM saturates the QFIM element by element.  The
normalization in Eq.~\eqref{eq:S_normalized_fisher} removes the diagonal
efficiencies and exposes only the delay mixing.  The finite-shot simulations
retain both the diagonal losses and the complete multinomial covariance.

\subsection{First- and second-order finite-copy scaling}
\label{sec:S_signal_scaling}

The local signal carried by a linear tangent is proportional to
\(\delta w\).  Its squared signal-to-noise ratio under independent sampling
is therefore controlled by
\begin{equation}
    R_1=N_{\rm shot}(\delta w)^2.
    \label{eq:S_R1}
\end{equation}
The leading product tangent is second order in the input amplitude, so its
corresponding local resource is
\begin{equation}
    R_2=N_{\rm shot}(\delta w)^4.
    \label{eq:S_R2}
\end{equation}
These relations concern the local signal order and are distinct from an
assumption of Gaussian feature noise.  We test them using exact multinomial
sampling from Eq.~\eqref{eq:S_shadow_povm}, while holding the reservoirs,
POVM, targets, and regression procedure fixed.

\begin{figure}[H]
    \centering
    \includegraphics[width=0.99\linewidth]{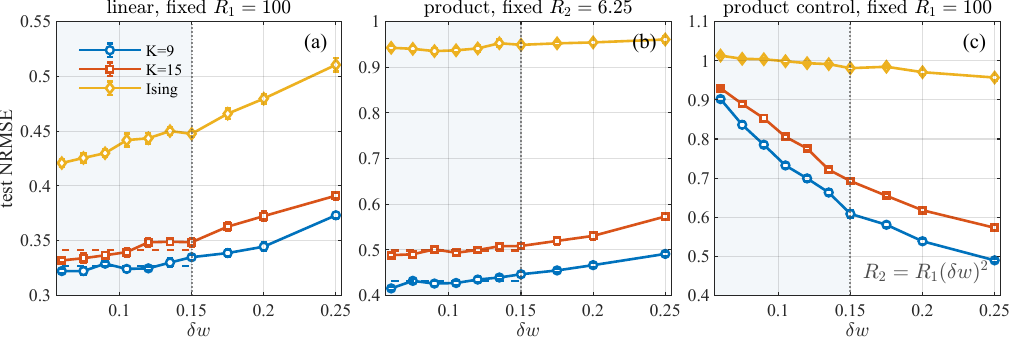}
    \caption{Finite-copy scaling at linear delay \(q=3\) and product delay
    \(q=3\).  (a) Linear prediction at fixed \(R_1=100\).  (b) Product
    prediction at fixed \(R_2=6.25\).  (c) The same product target at fixed
    \(R_1=100\), for which \(R_2=R_1(\delta w)^2\) is not held constant.  The
    shaded interval marks the resolved local regime
    \(0.06\le\delta w\le0.15\); dashed segments show the mean Clifford scores
    over that interval.  Error bars are standard errors over three input
    sequences and five multinomial realizations.}
    \label{fig:S_signal_scaling}
\end{figure}

At fixed \(R_1\), the linear curves form approximate plateaus in the local
regime.  At fixed \(R_2\), the Clifford product curves show the analogous
second-order plateaus.  By contrast, fixing \(R_1\) for the product task
forces \(R_2\) to decrease quadratically with \(\delta w\), producing the
strong drift in Fig.~\ref{fig:S_signal_scaling}(c).  Departures from a perfect
plateau at larger amplitudes arise from higher-order response terms and the
finite trajectory used to fit the readout.

\subsection{QFIM scaling with reservoir size}
\label{sec:S_qfim_scaling}

To compare memory windows of different lengths, we use the extremal
eigenvalues \(\lambda_{\min}(\mathsf H)\) and
\(\lambda_{\max}(\mathsf H)\), the total information \(\Tr\mathsf H\),
and
\begin{equation}
    r_{\rm eff}(\mathsf H)
    =\frac{(\Tr\mathsf H)^2}{\Tr(\mathsf H^2)},
    \qquad
    \operatorname{cond}(\mathsf H)
    =\frac{\lambda_{\max}(\mathsf H)}{\lambda_{\min}(\mathsf H)}.
    \label{eq:S_qfim_screening_metrics}
\end{equation}
These quantities describe complementary aspects of local delay
reconstruction.  A larger \(\Tr\mathsf H\) means greater total sensitivity
to the input history, but does not show how evenly that sensitivity is
distributed.  The smallest eigenvalue controls the least-resolved input
combination and is therefore a worst-case measure; larger
\(\lambda_{\min}\) is favorable.  We plot \(K\lambda_{\min}\) to remove the
trivial \(1/K\) scale associated with comparing windows of different lengths.
The largest eigenvalue measures the best-resolved combination, but a large
\(\lambda_{\max}\) alone may simply indicate that the information is
concentrated in a few directions.  Correspondingly,
\(r_{\rm eff}\in[1,K]\) estimates the number of appreciably resolved
directions, so values near \(K\) are favorable, whereas
\(\operatorname{cond}(\mathsf H)\ge1\) measures spectral imbalance and is
best when close to one.  Thus the favorable trends in
Fig.~\ref{fig:S_qfim_scaling} are large \(\Tr\mathsf H\),
\(K\lambda_{\min}\), and \(r_{\rm eff}/K\), together with a small condition
number.  These are diagnostics of the local QFI model rather than universal
predictors of every nonlinear task.

We next compare the deepest Singer window \(K=2^N+1\) with a random-Ising
ensemble for \(N=2,3,4,5\).  For the Ising reservoirs, \(K\) is a common
evaluation window rather than an intrinsic period.  Restricting both models
to the same \(K\) historical inputs gives QFIMs of the same dimension and
therefore compares their ability to resolve the same reconstruction problem.
The Ising Hamiltonian and parameter ranges are those in
Eqs.~\eqref{eq:S_ising_U} and \eqref{eq:S_ising_H}.  For each \(N\), we retain
\(1000\) samples whose fixed-point residual is below \(10^{-9}\), whose fixed
state is full rank to numerical tolerance, and whose nonstationary spectral
radius is below \(1-10^{-5}\).  The \(K\)-delay QFIM is then evaluated
analytically from the channel derivative at \(s_0=1/2\).
Figure~\ref{fig:S_qfim_scaling} compares the resulting Ising distributions
with the optimized Singer construction; the analysis below explains the
scaling of the Singer points.

For reproducibility, Table~\ref{tab:S_singer_scaling_models} gives one
explicit realization of every Singer curve in
Figs.~\ref{fig:S_qfim_scaling} and~\ref{fig:S_qfim_window_scan}.  For
\(N=2,3\), the Pauli strings refer to the matrices \(M_5\) and \(M_9\) in
Eqs.~\eqref{eq:S_explicit_M5} and~\eqref{eq:S_explicit_M9}.  For
\(N=4,5\), we construct multiplication by an element of order \(K\) in a
polynomial basis of \(\mathbb F_{2^{2N}}\) and transform its invariant
alternating form to the symplectic form in Eq.~\eqref{eq:S_symplectic_form}.
This fixes reproducible binary matrices \(M_{17}\) and \(M_{33}\); their
Clifford lifts are specified below through their action on all single-qubit
Pauli generators.  The listed pairs
produce the selected hit words through Eq.~\eqref{eq:S_eta_word}, and the
fading values obey \(r^2=(h-1)/h\).  A change of symplectic basis changes the
literal Pauli strings but gives a Clifford-equivalent channel with the same
hit word and QFIM.

\begin{table}[!t]
    \centering
    \begin{tabular}{ccccccc}
        \toprule
        \(N\) & \(K\) & route & \(Q\) & \(G\) & \(h\) & \(r\) \\
        \midrule
        2 & 5  & \(V_5\)  & \(YI\)    & \(XI\)    & 2  & \(1/\sqrt2\) \\
        3 & 9  & \(V_9\)  & \(IZI\)   & \(IYI\)   & 2  & \(1/\sqrt2\) \\
        4 & 17 & \(V_{17}\) & \(IIIZ\)  & \(IIIX\)  & 6  & \(\sqrt{5/6}\) \\
        5 & 33 & \(V_{33}\) & \(IIIYI\) & \(IIXXI\) & 12 & \(\sqrt{11/12}\) \\
        \bottomrule
    \end{tabular}
    \caption{Representative Singer reservoirs used in the QFIM scaling and
    history-window comparisons.  The qubit order follows the convention of
    Sec.~S4.  All four pairs satisfy \(\{Q,G\}=0\) and the global hitting
    condition.}
    \label{tab:S_singer_scaling_models}
\end{table}

\begin{figure}[H]
    \centering
    \includegraphics[width=0.94\linewidth]{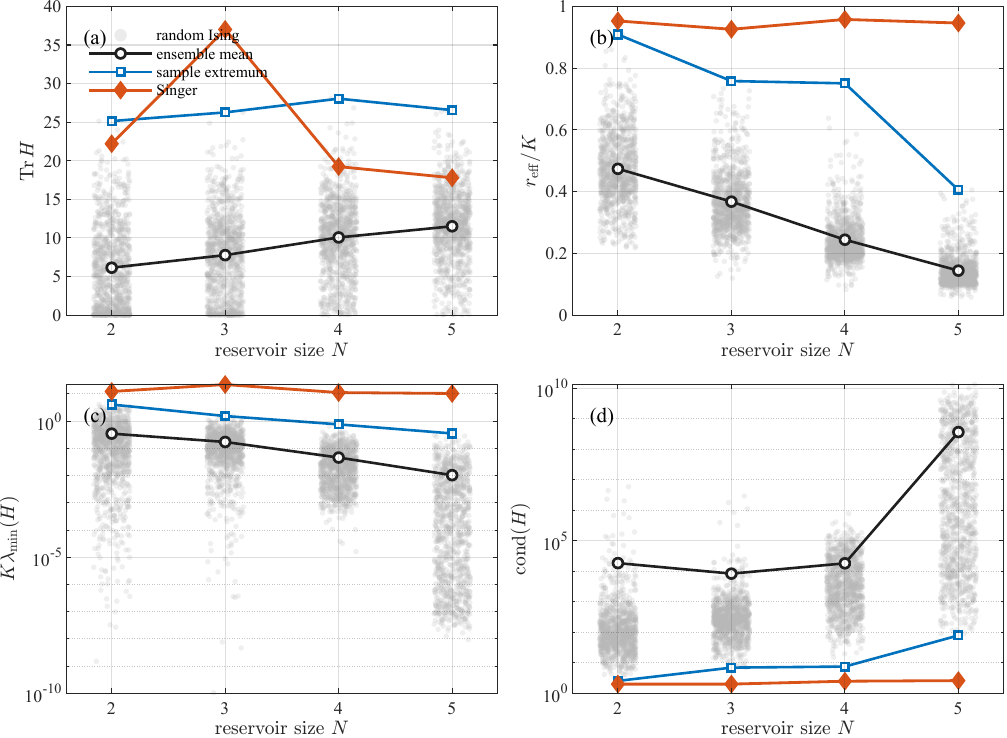}
    \caption{QFIM scaling for Singer reservoirs and ESP-qualified random
    Ising reservoirs.  Gray points show \(1000\) Ising samples for each \(N\);
    black circles show their means, and blue squares show the sample maxima
    in panels (a)--(c) and the sample minimum in panel (d).  Orange diamonds
    show minimum-hit Singer constructions, with \(r\) chosen according to
    Eq.~\eqref{eq:S_singer_deep_qfi_optimum}; the corresponding \(Q,G,r\)
    values are listed in Table~\ref{tab:S_singer_scaling_models}.  The four panels report
    (a) total Fisher information,
    (b) relative effective rank, (c) the weakest-delay scale, and (d) the
    condition number of full-rank QFIMs.}
    \label{fig:S_qfim_scaling}
\end{figure}

Figure~\ref{fig:S_qfim_scaling} separates the total sensitivity to the input
history from its distribution over delay directions.  Although
\(\Tr\mathsf H\) alone does not show a uniform ordering at every reservoir
size, the Singer construction retains a relative effective rank close to one,
a finite weakest-delay scale, and an order-one condition number.  The random
Ising ensemble instead shifts toward lower relative rank and weaker least
resolved directions as \(N\) increases, while its condition-number
distribution broadens.  Hence comparable total QFI does not imply comparable
resolution of the full input history.  The blue extremum is selected
separately in each panel and need not correspond to a single Ising reservoir;
it is included to show the best sampled value of each diagnostic rather than a
simultaneous optimum.

The common cutoff \(K=2^N+1\) in Fig.~\ref{fig:S_qfim_scaling} is not meant
to assign a period to the Ising dynamics.  To verify that the comparison is
not an artifact of this single endpoint, we fix \(N=4,5\) and evaluate every
leading delay window

\begin{equation}
    \mathsf H_L=\mathsf H_{1:L,1:L},
    \qquad L=1,\ldots,2^N+1,
    \label{eq:S_common_history_window}
\end{equation}

for the same \(1000\) Ising reservoirs.  Each \(\mathsf H_L\) describes the
same reconstruction problem for both model classes: resolving the most recent
\(L\) inputs from the final reservoir state.  \textbf{Here \(L\) specifies the
dimension of the reconstruction task; it is not a presumed period of the
Ising dynamics.}  The Singer reservoirs and all Ising samples are held fixed
as \(L\) is varied; neither class is reoptimized for an individual window.
In particular, the \(N=4,5\) Singer curves use the corresponding fixed
\(Q,G,r\) entries in Table~\ref{tab:S_singer_scaling_models} throughout the
scan.

\begin{figure}[H]
    \centering
    \includegraphics[width=0.96\linewidth]{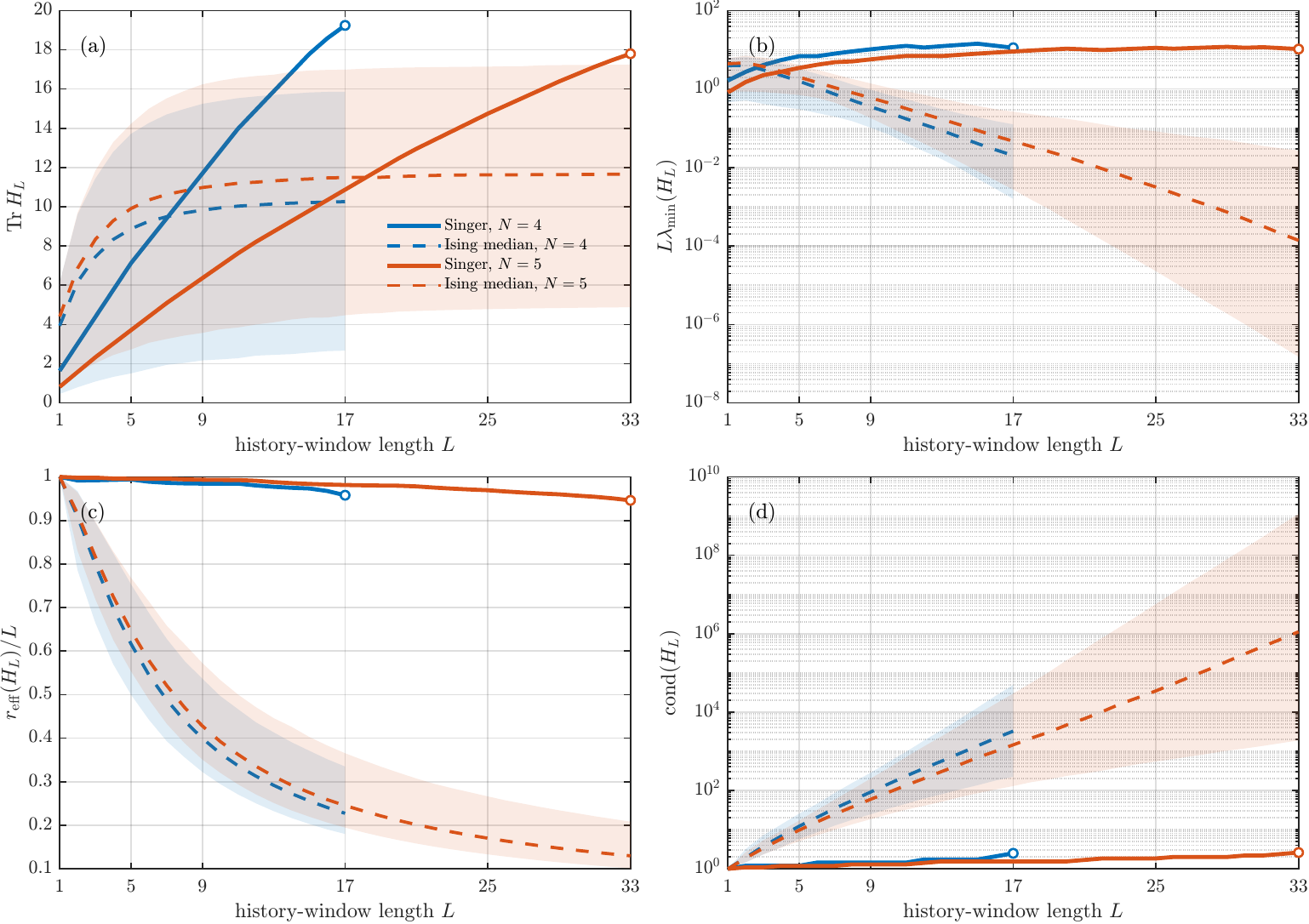}
    \caption{QFIM diagnostics as the common history window \(L\) is enlarged
    at fixed reservoir sizes \(N=4\) and \(N=5\).  Solid curves show the
    fixed Singer constructions listed in
    Table~\ref{tab:S_singer_scaling_models}.  Dashed curves and shaded bands show, respectively,
    the median and 10th--90th percentile range of the same \(1000\) ESP-qualified
    Ising reservoirs used in Fig.~\ref{fig:S_qfim_scaling}.  The panels report
    (a) total QFI, (b) \(L\lambda_{\min}\), (c) relative effective rank, and
    (d) condition number.  All sampled QFIMs are full rank at the numerical
    tolerance used here.}
    \label{fig:S_qfim_window_scan}
\end{figure}

Figure~\ref{fig:S_qfim_window_scan} shows what is hidden by an endpoint-only
comparison.  The Ising median \(\Tr\mathsf H_L\) initially grows
but soon saturates, whereas the added delay directions become progressively
weaker and more uneven.  Consequently,
\(L\lambda_{\min}(\mathsf H_L)\) and \(r_{\rm eff}(\mathsf H_L)/L\) decrease
with \(L\), while the condition number rises rapidly.  The Singer QFIM instead
adds one orthogonal direction at each step.  Its relative effective rank stays
close to one, its weakest-direction scale remains finite, and its condition
number remains of order unity throughout the designed window.  Thus the
advantage at \(L=K\) is the endpoint of a systematic window-length trend, not
a consequence of imposing the Singer period on the Ising model.

The Singer scaling follows in two steps.  We first optimize the QFIM spectrum
at a fixed number of damping hits and then use the Singer hit-word
distribution to determine its dependence on the reservoir size.

\paragraph{Singer spectrum at fixed hit count.}
The Singer scaling can be obtained directly from
Eq.~\eqref{eq:S_many_H_eta}.  Let
\(h=\sum_{j=1}^K\eta_j\ge2\) be the number of damping hits in one period and
cyclically choose the word so that \(\eta_K=1\).  The largest and smallest
diagonal entries are then
\begin{equation}
    \lambda_{\max}(r)=\pi^2(1-r^2),
    \qquad
    \lambda_{\min}(r)=\pi^2(1-r^2)r^{2(h-1)}.
    \label{eq:S_singer_extreme_qfi}
\end{equation}
Optimizing the deepest-delay QFI over the fading parameter gives
\begin{equation}
    r_\ast^2=\frac{h-1}{h},
    \qquad
    \lambda_{\min}^\ast
    =\frac{\pi^2}{h}\left(1-\frac1h\right)^{h-1},
    \qquad
    \operatorname{cond}(\mathsf H_\ast)
    =\left(1-\frac1h\right)^{-(h-1)}<e.
    \label{eq:S_singer_deep_qfi_optimum}
\end{equation}
At this optimum, every QFIM eigenvalue lies between
\(\lambda_{\min}^\ast>\pi^2/(eh)\) and
\(\lambda_{\max}^\ast=\pi^2/h\).  Summing the \(K\) eigenvalues gives the
trace bounds below.  The effective-rank bound follows without introducing
their individual multiplicities.  If the eigenvalues are \(\lambda_k\),
then
\(\sum_k\lambda_k^2\le\lambda_{\max}\sum_k\lambda_k\), and hence
\[
    r_{\rm eff}
    =\frac{(\sum_k\lambda_k)^2}{\sum_k\lambda_k^2}
    \ge\frac{\sum_k\lambda_k}{\lambda_{\max}}
    \ge K\frac{\lambda_{\min}}{\lambda_{\max}}.
\]
Together with \(r_{\rm eff}\le K\), this yields
\begin{equation}
    \frac{\pi^2K}{eh}
    <K\lambda_{\min}^\ast
    \le\Tr\mathsf H_\ast
    \le\frac{\pi^2K}{h},
    \qquad
    \frac{K}{e}<r_{\rm eff}(\mathsf H_\ast)\le K,
    \qquad
    2\le\operatorname{cond}(\mathsf H_\ast)<e.
    \label{eq:S_singer_metric_scaling}
\end{equation}
Thus
\(\Tr\mathsf H_\ast,K\lambda_{\min}^\ast=\Theta(K/h)\),
\(r_{\rm eff}(\mathsf H_\ast)=\Theta(K)\), and
\(\operatorname{cond}(\mathsf H_\ast)=O(1)\).  In particular, increasing the
route period does not make the designed QFIM exponentially ill conditioned.

\paragraph{Minimum hit count and its scaling.}
The size dependence of the Singer spectrum is therefore controlled by the
number \(h\) of damping hits in one period.  The transparent-run constraint in
Sec.~\ref{sec:S_esp} first gives the elementary bound
\(h\ge\lceil K/(2N)\rceil\).  To see this, sum the \(K\) cyclic
length-\(2N\) window inequalities established in
Sec.~\ref{sec:S_esp}; every hit is then counted exactly \(2N\) times, giving
\(2Nh\ge K\).  This elementary bound uses only the spanning property of
consecutive Singer labels and does not rely on a finite-field description.

The elementary argument rules out very sparse hit words but is not tight.
The Singer hit words form a well-known binary cyclic code whose exact weight
distribution is known \cite{SMSchoofVanDerVlugt1991Hecke}.  For every nonzero
word and \(N\ge3\), that result places the hit count within \(2^{N/2}\) of
half the period.  In the present notation,
\begin{equation}
    \left|h-\frac{K}{2}\right|<2^{N/2},
    \qquad
    h=\frac{K}{2}+O(\sqrt K),
    \qquad
    \frac{K}{h}=2+O(K^{-1/2}).
    \label{eq:S_singer_hit_concentration}
\end{equation}
Thus every nonzero Singer hit word has asymptotic density one half.  The cited
weight distribution also determines the exact finite-size minimum; direct
enumeration provides an independent check for the sizes studied below.

Combining Eq.~\eqref{eq:S_singer_hit_concentration} with
Eq.~\eqref{eq:S_singer_metric_scaling} gives the dimension-explicit Singer
scaling
\begin{equation}
    \Tr\mathsf H_\ast, K\lambda_{\min}^\ast=\Theta(1),
    \qquad
    r_{\rm eff}(\mathsf H_\ast)=\Theta(2^N),
    \qquad
    \operatorname{cond}(\mathsf H_\ast)=O(1).
    \label{eq:S_singer_N_scaling}
\end{equation}
The weaker elementary inequality \(h\ge K/(2N)\) by itself gives only
\(\Tr\mathsf H_\ast=O(N)\) and
\(K\lambda_{\min}^\ast=O(N)\).  The cyclic-code result shows that Singer hit
words actually satisfy \(h=\Theta(K)\), sharpening both estimates to
\(\Theta(1)\).

\paragraph{Comparison with random Ising reservoirs.}
The reference construction uses a minimum-hit Singer word, because
\(\lambda_{\min}^\ast\) decreases with \(h\), and then selects, among tied
words, the rotation with the largest effective rank.  The exact cyclic-code
spectrum, supplemented by direct enumeration for \(N=2\), gives
\(h=2,2,6,12\) for \(N=2,3,4,5\), respectively.  Hence the deepest-delay QFI
decreases as \(1/h\), while the compensated quantity
\(K\lambda_{\min}\) follows the \(K/h\) scale in
Eq.~\eqref{eq:S_singer_metric_scaling}.

For \(N=2,\ldots,5\), the Singer routes retain nearly uniform QFIM spectra:
\(r_{\rm eff}/K\) remains above \(0.92\),
\(K\lambda_{\min}(\mathsf H)\) above \(10\), and
\(\operatorname{cond}(\mathsf H)\) below \(2.6\).  By contrast, the
random-Ising ensemble exhibits a decreasing relative effective rank and
weakest-delay information, together with a rapidly broadening
condition-number distribution.  Although this finite-size comparison does
not establish the asymptotic behavior of random Ising reservoirs, it shows
that designed routing maintains a well-resolved delay space as the memory
window grows.

\paragraph{Clifford generators used in the scaling analysis.}
The generators \(V_5\) and \(V_9\) are given in
Eqs.~\eqref{eq:S_explicit_V5} and \eqref{eq:S_explicit_V9}.  To make the
scaling results fully reproducible, we specify the \(N=4,5\) Clifford
generators used in the numerical curves through their action on the
single-qubit Pauli generators, choosing positive Pauli phases.  Up to an
irrelevant global phase,
\begin{equation}
\begin{aligned}
 V_{17}X_1V_{17}^\dagger&=ZXZI,&
 V_{17}Z_1V_{17}^\dagger&=XYXZ,\\
 V_{17}X_2V_{17}^\dagger&=ZYXX,&
 V_{17}Z_2V_{17}^\dagger&=XYII,\\
 V_{17}X_3V_{17}^\dagger&=YXXI,&
 V_{17}Z_3V_{17}^\dagger&=XIXZ,\\
 V_{17}X_4V_{17}^\dagger&=XIYX,&
 V_{17}Z_4V_{17}^\dagger&=IYYY,
\end{aligned}
\label{eq:S_explicit_V17_action}
\end{equation}
and
\begin{equation}
\begin{aligned}
 V_{33}X_1V_{33}^\dagger&=IYXZI,&
 V_{33}Z_1V_{33}^\dagger&=ZYZZX,\\
 V_{33}X_2V_{33}^\dagger&=IYZXZ,&
 V_{33}Z_2V_{33}^\dagger&=IIXIZ,\\
 V_{33}X_3V_{33}^\dagger&=ZXYZX,&
 V_{33}Z_3V_{33}^\dagger&=ZZYZX,\\
 V_{33}X_4V_{33}^\dagger&=ZYZYY,&
 V_{33}Z_4V_{33}^\dagger&=XIXZI,\\
 V_{33}X_5V_{33}^\dagger&=IYZYY,&
 V_{33}Z_5V_{33}^\dagger&=XYZYX.
\end{aligned}
\label{eq:S_explicit_V33_action}
\end{equation}
These images preserve the Pauli commutation relations and therefore define
valid Clifford unitaries.  Together with
Table~\ref{tab:S_singer_scaling_models}, they fully specify the Singer
channels in Figs.~\ref{fig:S_qfim_scaling}
and~\ref{fig:S_qfim_window_scan}.

\subsection{Delayed-product memory along minimum-hit Singer routes}
\label{sec:S_singer_product_scaling}

Section~\ref{sec:S_route_specific_second_order} defined the consecutive
second-order write-in window \(L_2\) and proved the general bound
\(L_2\le2N-2\).  At fixed route period, the first- and second-order objectives
are not identical.  Minimizing the number of damping hits \(h\) maximizes the
weakest first-order QFI after optimizing \(r\), whereas \(L_2\) measures the
longest initial range of delayed products with nonzero leading mixed response.
We therefore ask whether a minimum-hit Singer route can simultaneously retain
a long consecutive delayed-product window.

For \(Q,G\), and a Singer generator \(M\),
\(L_2(Q,G;M)\) is the largest integer for which
\(w_{\rm prod}(\ell+1;Q,G)=1\) for every
\(\ell=1,\ldots,L_2\).  Equivalently, all product delays
\(q=2,\ldots,L_2+1\) have a nonzero leading mixed response.  We examine this
compatibility over all inequivalent Singer-generator classes rather than for
a single reference route.

For a fixed Singer generator \(M\), vary the legal Pauli strings \(Q\ne I\)
and \(G\) with \(\{Q,G\}=0\).  At an attainable hit count \(h\), define
\begin{equation}
    L_2^\star(h;M)
    =
    \max_{\substack{Q\ne I,\ \{Q,G\}=0\\
          \sum_{j=1}^{K}\eta_j(Q,G;M)=h}}
    L_2(Q,G;M).
    \label{eq:S_hit_resolved_product_window}
\end{equation}
This is the longest consecutive delayed-product window compatible with that
amount of damping.  The attainable Singer hit counts are even, so
the scan proceeds through \(h_{\min},h_{\min}+2,\ldots\), rather than through
every integer.  This parity follows because the sum modulo two of a complete
hit word vanishes.

\textbf{Class-resolved enumeration.}  For each \(N\), we choose one Singer
generator \(M_{\rm ref}^{(N)}\) only as a starting point.  Every primitive
power \((M_{\rm ref}^{(N)})^p\), with \(\gcd(p,K)=1\), is included.  Powers
related by \(p\mapsto2^jp\pmod K\) have the same binary characteristic
polynomial and differ only by the associated field-coordinate relabeling, so
one representative of each such doubling orbit is retained
\cite{SMHestenes1970SingerGroups}.
For every representative, we enumerate all nonidentity \(Q\) and all \(G\)
satisfying \(\{Q,G\}=0\), and maximize \(L_2\) at \(h=h_{\min}\).  If a class
does not attain \(N-1\) there, we scan its higher admissible hit counts and
record the smallest increase that reaches \(N-1\).  This is an exact
enumeration of both the Pauli pairs and the inequivalent Singer-generator
classes.

To summarize the dependence on the Singer-generator class, we record four
quantities.  The values \(L_{2,\min}^{\rm cls}\) and
\(L_{2,\max}^{\rm cls}\) are the smallest and largest optimized windows at
\(h=h_{\min}\) across all inequivalent classes.
The quantity \(\Delta h_{\max}^{\rm cls}\) is the largest additional hit count
required by any class to reach \(L_2=N-1\), while
\(n_{\rm opt}/n_{\rm cls}\) gives the fraction of classes that already reach
this value at \(h_{\min}\).  These quantities are independent of the Singer
generator chosen as the initial reference.

\begin{table}[H]
    \centering
    \begin{tabular}{@{}cccccc@{}}
        \toprule
        \(N\) & \(h_{\min}\)
        & \(L_{2,\min}^{\rm cls}\)
        & \(L_{2,\max}^{\rm cls}\) & \(\Delta h_{\max}^{\rm cls}\)
        & \(n_{\rm opt}/n_{\rm cls}\) \\
        \midrule
        2  & 2   & 1  & 1  & 0 & 1/1   \\
        3  & 2   & 2  & 2  & 0 & 1/1   \\
        4  & 6   & 3  & 3  & 0 & 2/2   \\
        5  & 12  & 4  & 4  & 0 & 2/2   \\
        6  & 26  & 5  & 5  & 0 & 4/4   \\
        7  & 54  & 6  & 6  & 0 & 6/6   \\
        8  & 114 & 7  & 7  & 0 & 16/16 \\
        9  & 234 & 7  & 8  & 2 & 17/18 \\
        10 & 482 & 9  & 9  & 0 & 40/40 \\
        11 & 980 & 10 & 10 & 0 & 62/62 \\
        \bottomrule
    \end{tabular}
    \caption{Exact class-resolved delayed-product windows for
    \(N=2,\ldots,11\).  The third and fourth columns give the smallest and
    largest \(L_2^\star(h_{\min};M)\) over all inequivalent Singer-generator
    classes.  The fifth column is the largest additional hit count required
    by any class to reach \(L_2=N-1\).  The final column gives the number
    \(n_{\rm opt}\) of classes already attaining \(N-1\) at \(h_{\min}\),
    divided by the total number \(n_{\rm cls}\) of classes.}
    \label{tab:S_hit_resolved_product_window}
\end{table}

At \(N=9\), one of the 18 period-513 generator classes is exceptional.  A
representative of that class is the original reference generator, for which
\begin{equation}
    L_2^\star(234;M_{\rm ref}^{(9)})=7,
    \qquad
    L_2^\star(236;M_{\rm ref}^{(9)})=8.
    \label{eq:S_N9_hit_tradeoff}
\end{equation}
Thus the next admissible hit count restores \(N-1=8\) for the exceptional
class, giving \(\Delta h_{\max}^{\rm cls}=2\).  The other 17 period-513 classes
already attain \(L_2=8\) at \(h_{\min}=234\).  At every other tested size,
all inequivalent classes attain \(L_2=N-1\) at the minimum hit count.

\textbf{Within the tested range, the first-order choice \(h=h_{\min}\) is
therefore compatible with a consecutive delayed-product window of length
\(N-1\), apart from one ordering class at \(N=9\).}  This is an exact
finite-size statement, not a theorem that strengthens the general bound in
Eq.~\eqref{eq:S_product_initial_run_bound}.  Moreover, \(h\) and \(L_2\)
describe whether the response is available; the final task NRMSE also depends
on response amplitudes and readout noise.

\makeatletter
\let\SM@saved@innerbib\auto@bib@innerbib
\let\SM@saved@frontmatternotes\@FMN@list
\let\auto@bib@innerbib\@empty
\let\@FMN@list\@empty
\makeatother

\makeatletter
\let\auto@bib@innerbib\SM@saved@innerbib
\let\@FMN@list\SM@saved@frontmatternotes
\makeatother

\end{document}